\documentclass[11pt,a4paper]{article}

\usepackage[truedimen,margin=34mm]{geometry} 

\usepackage{mathrsfs}
\usepackage{amssymb}
\usepackage{amsmath}
\usepackage{ascmac}
\usepackage{amsthm}
\usepackage[pdftex]{graphicx}
\usepackage[pdftex]{color}
\usepackage{natbib}
\usepackage{setspace}

\usepackage{comment}
\usepackage{bm}

\usepackage{titlesec}
\titleformat*{\section}{\large\bfseries}
\titleformat*{\subsection}{\it}

\newtheorem{thm}{Theorem}
\newtheorem{lem}{Lemma}

\newtheorem{prp}{Proposition}

\newtheorem{algo}{Algorithm}

\newtheorem{remark}{Remark}[section]
\def\ep{{\varepsilon}}

\def\th{{\theta}}

\def\Ga{\Gamma}

\def\thh{{\widehat{\theta}}}

\def\ta{{\tau}}
\def\pit{{\tilde \pi}}

\def\al{{\alpha}}
\def\be{{\beta}}
\def\ga{{\gamma}}

\def\om{{\omega}}
\def\ka{{\kappa}}

\def\non{{\nonumber}}

\title{{\bf Shrinkage with Robustness:\\
Log-Adjusted Priors for Sparse Signals}}

\date{}
\author{}

\begin{document}

\maketitle
\doublespacing

\vspace{-1.5cm}
\begin{center}
{\large  Yasuyuki Hamura$^1$, Kaoru Irie$^2$ and Shonosuke Sugasawa$^3$}
\end{center}

\noindent
$^1$Graduate School of Economics, The University of Tokyo\\
$^2$Faculty of Economics, The University of Tokyo\\
$^3$Center for Spatial Information Science, The University of Tokyo

\vspace{5mm}
\begin{center}
{\bf \large Abstract}
\end{center}
We introduce a new class of distributions named log-adjusted shrinkage priors for the analysis of sparse signals, which extends the three parameter beta priors by multiplying an additional log-term to their densities. 
The proposed prior has density tails that are heavier than even those of the Cauchy distribution and realizes the tail-robustness of the Bayes estimator, while keeping the strong shrinkage effect on noises. 
We verify this property via the improved posterior mean squared errors in the tail. An integral representation with latent variables for the new density is available and enables fast and simple Gibbs samplers for the full posterior analysis. Our log-adjusted prior is significantly different from existing shrinkage priors with logarithms for allowing its further generalization by multiple log-terms in the density. 
The performance of the proposed priors is investigated through simulation studies and data analysis.


\newpage
\section{Introduction}

Developing new classes of continuous prior distributions that realize the shrinkage effect of variable-selection type on location parameters has been an important research top in the last few decades, especially in the context of the analysis of high-dimensional datasets to properly express one's prior belief on ``few large signals among noises''. As pointed out by \cite{carvalho2009handling,carvalho2010horseshoe}, we can express such belief explicitly via the parameterization of shrinkage effect in the Bayes estimator that shrinks the observed signals to zero or baseline. This parametrization opens the path to crafting the new class of continuous priors that mimic the discrete mixture for variable selection or the spike-and-slab priors \citep{ishwaran2005spike}, which is more desirable in the high-dimensional context than the existing shrinkage priors \citep[e.g.][]{strawderman1971proper,berger1980robust,park2008bayesian}. 
The desirable prior here should, of course, shrink the negligible noises toward zero, but also be robust to outlying large signals in the sense that such signals are kept unshrunk in the posterior analysis. 
The latter property is typically called tail-robustness \citep[e.g.][]{carvalho2010horseshoe}, and the aim of this research is define a new class of shrinkage priors with strong tail-robustness.

The aforementioned parametrization describes both shrinkage effect and tail-robustness implicitly assumed in the prior of interest. Suppose we observe $y_i \sim N(\theta _i , 1 )$ independently for $i=1,\ldots,n$ and the prior is given by $\theta _i \sim N(0,\tau u_i)$ (and $\tau = 1$, for simplicity) and $u_i\sim \pi (u_i)$. 
Then, the Bayes estimator of true signal $\theta _i$ is written as $(1-E[\kappa _i | y_i]) y_i$, where $\kappa _i = 1 / (1+u_i)$. It is this parameter, $\kappa _i$, that controls the amount of shrinkage in the Bayes estimator. In the presence of sparse signals, the standard choice of priors has been the beta distribution \citep{armagan2011generalized,perez2017scaled}, originated from the half-Cauchy distribution (\citealt{gelman2006prior}, or horseshoe prior; \citealt{carvalho2009handling,carvalho2010horseshoe}) given by $\pi (\kappa_i ) \propto \kappa_i^{b-1} (1-\kappa_i )^{a-1}$, $\kappa_i \in (0,1)$, with positive $a$ and $b$. The appropriate modeling of shrinkage and robustness is then translated into the choice of extremely small shape parameters $(a,b)$. This preference on the choice of hyperparameters is, however, against the finding of \cite{bai2019large}; to guarantee the desirable posterior concentration for both small and large signals, $a$ can be extremely small ($a=1/n$) but $b$ must be sufficiently large ($b\ge 1/2$), which clarifies the limitation of the class of beta distributions. 

In this research, we consider the extension of beta-type shrinkage priors to strengthen the prior tail-robustness. Specifically, we propose the following modified version of the beta prior: 
\begin{equation}\label{EHS}
\pi (\kappa_i ) \propto \kappa_i^{b-1} (1-\kappa_i )^{a-1} \left( 1 - \log \kappa_i  \right) ^{-(1+\gamma )}, 
\end{equation}
where $\gamma > 0$ is a newly introduced hyperparameter. The use of logarithm in the density slightly ``slows down'' the divergence of the density as $\kappa _i \downarrow 0$ and, in fact, makes the density kernel above integrable even if $b=0$, as shown in Theorem~\ref{thm:prop}. This distribution allows the stronger tail-robustness than the beta prior by setting $b = 0$, while remaining in the class of proper priors. 

The use of logarithm term in the density function to define the new class of distributions has motivated many research on posterior inference. They include the analysis of ultra-sparse signals \citep{bhadra2017horseshoe+}, robust regression \citep{Gagnon2019}, and admissibility \citep{maruyama2019admissible}. The shrinkage with robustness-- our research goal-- has also been considered in \cite{womack2019heavy} as the heavy-tailed extension of the horseshoe prior. A similar log-adjusted method was employed in \cite{hamura2019global} for the analysis of high-dimensional counts. 
However, the proposed prior has a different functional form from those listed above, and the superiority of the proposed prior to existing ones will be demonstrated via improvement of the mean squared error for large $y_i$, as summarized in Theorem~\ref{thm:MSE1}. 

Another advantage of the prior of our interest is the potential of further generalization, by which one may modify the proper prior ``as robust as possible''. Although the prior in (\ref{EHS}) becomes improper with $\ga = 0$, we can multiply another log-term as 
\begin{equation*}
\pi (\kappa_i ) \propto \kappa_i ^{b-1} (1-\kappa_i )^{a-1} \left( 1 - \log \kappa_i \right) ^{-1} \left\{ 1 + \log (1 - \log \kappa_i ) \right\} ^{-(1+\gamma )},
\end{equation*}
which is proper again even if $b=0$ as long as $\gamma>0$. 
Notably, we can repeatedly iterate this process of extension; if $\gamma = 0$ in the above equation and the density becomes improper, then the reciprocal of another log-term, $1+\log \{ 1+\log ( 1-\log \kappa_i ) \}$, can be multiplied to the density to regain the proper prior. It is expected, and verified later in Theorem~\ref{thm:MSE2}, that such extension provides the stronger tail-robustness and makes the choice of $\ga$ less sensitive to the posterior analysis.

The Bayes estimator under the proposed priors has no closed form, even if global scale $\tau$ is fixed, due to the intractable normalizing constant. Yet, the estimator can be evaluated fast by simulation. It is shown that the prior density admits the integral representation, or the augmentation by latent variables that follow gamma-shape Markov processes, by which the full conditional posteriors of those parameters and latent variables become normal, (inverse) gamma or generalized inverse Gaussian distributions. Sampling from those distributions is trivial, and the full posterior analysis becomes available by the simple but efficient Gibbs sampler.

The rest of this paper is organized as follows. In Section \ref{sec:EHS}, we define the log-adjusted shrinkage prior and its extension, and provide the theoretical properties, the improvement of the mean squared errors of Bayes estimators, and the Gibbs sampler by augmentation. 
Simulation studies and data analysis follow in Section \ref{sec:numerical_study}, with the extensive comparative analysis with the existing shrinkage priors. 
We conclude our paper in Section \ref{sec:discussion} with the further discussion on the possibility of justifying the log-adjusted prior as the continuous alternative of variable selection that approximates spike-and-slab priors. 

All proofs and technical details are given in the Appendix.

\section{Log-Adjusted Shrinkage Priors}
\label{sec:EHS}

\subsection{The proposed prior and its properties}
Suppose we observe an $n$-dimensional vector $(y_1,\ldots,y_n)$, that $y_i|\theta_i\sim N(\theta_i,1)$ independently for $i=1,\ldots,n$. To estimate signals $(\theta_1,\ldots,\theta_n)$ that are potentially sparse, we adopt locally adaptive shrinkage priors known as global-local shrinkage priors \citep{polson2012local,polson2012half,bhadra2016default} given by 
\begin{equation}\label{GLSP}
\theta_i|\tau,u_i\sim N(0,\tau u_i) \ \ \ \mathrm{and} \ \ \  u_i\sim\pi(u_i), \ \ \ \mathrm{for} \ i=1,\ldots,n,
\end{equation}
where both $\tau$ and $(u_1,\ldots ,u_n)$ are all positive. Here, $\tau$ is the global shrinkage parameter that shrinks all $\theta_i$'s toward zero uniformly, while $u_i$ is the local scale parameters and customizes the shrinkage effect for each individual $i$. For simplicity, we assume $\ta = 1$ to focus our theoretical development on the priors for local scale parameters. 
We propose the following modified version of the scaled beta distribution: 
\begin{equation}\label{EHS-u}
\pi (u_i)=C(a,b,\gamma)^{-1}u_i^{a-1} (1+u_i )^{-(a+b)} \left\{ 1 + \log (1 + u_i)  \right\} ^{-(1+\gamma )},
\end{equation}
where $C(a,b,\gamma)$ is a normalizing constant. Note that the class of distributions defined by density (\ref{EHS-u}) includes the scaled beta distributions \citep{armagan2011generalized} as the density of $\gamma = -1$ and positive $a$ and $b$. The hyperparameters, $(a,b,\gamma )$, determine the functional form of the density around the origin and in the tails. Shape parameters $a$ and $b$ control shrinkage effect and tail robustness for Bayes estimators, respectively, and both parameters should be set to small values in order to achieve the desirable shrinkage and robustness properties. 
Specifically, we set $a=1/n$, following \cite{bai2019large}, to realize the strong shrinkage effect on noises and set $b=0$ to attain the strong tail robustness.
Note again that setting $b=0$ in the original scaled beta distribution leads to an improper prior, thereby it cannot be adopted as shrinkage priors in practice. 
The new parameter $\gamma$ also affects the tail behavior of the density as $b$ does, but it would have less impact on posterior analysis.
We may either fix $\gamma$ subjectively to a certain value, such as $\gamma=1$, or take the fully Bayesian approach by considering the prior for $\gamma$ as we discuss in the subsequent section. 
The new priors for $\theta_i$ under (\ref{GLSP}) with the log-adjusted scaled beta distribution (\ref{EHS-u}) is named {\it log-adjusted shrinkage priors}. 
In what follows, we demonstrate properties of the proposed prior with general hyperparameters, $(a, b, \gamma)$, but the priors of our interest and recommendation are those with $a=1/n$ and $b=0$.

We first provide important properties of the proposed shrinkage prior for $\theta_i$ in the theorem below.

\begin{thm}\label{thm:prop}
The log-adjusted shrinkage prior for $\theta_i$, $\pi (\theta _i)$, satisfies the following properties.
\begin{itemize}
\item[1.] $\pi (\theta _i)$ is proper if $a>0$, $b\ge 0$ and $\gamma >0$. 

\item[2.]
$\lim_{|\theta_i|\to 0}\pi (\theta_i)=\infty$ for $a \le 1 / 2$. 

\item[3.]
$\pi (\theta_i)\propto |\theta_i|^{- 2b -1}L( | \theta_i |)$ under $|\theta_i|\to\infty$, where $L( \cdot )$ is a slowly varying function satisfying $\lim_{M\to\infty}L(Mu)/L(M)=1$ for all $u>0$.

\end{itemize}
\end{thm}

\noindent 
It is notable that the prior is proper even if $b=0$ from the first property; this is obviously due to the additional log-term in (\ref{EHS-u}).
The second property is the same as that of the original beta prior, which indicates that the proposed prior has the density with the spike around the origin and holds strong shrinkage property.
It also means that the additional log-term does not change the shrinkage property of the original beta-type prior.
From the third property, our proposal of setting $b=0$ results in the extremely heavy tailed prior distribution for $\theta_i$, whose tail is heavier than even the Cauchy distribution. 
Such heavy-tailed properties are essential for strong-tail robustness, as shown in Theorem \ref{thm:MSE1}.

Figure~\ref{fig:dens} shows the examples of the log-adjusted beta distribution given in (\ref{EHS}), in the scale of $\ka _i = 1/(1+u_i)$, for different choices of hyperparameter $\ga$. When compared with beta density $Be(1/2,1/2)$, which is the half-Cauchy distribution in the scale of $u_i$ and realizes the horseshoe prior, the log-adjusted shrinkage densities have steeper spike as $\ka _i \to 0$, reflecting its tail property introduced by the additional log-term. The shrinkage effect is also affected by this additional term in the density, but the densities with moderate values of $\ga$, such as $\ga = 0.5$ or $\ga =1$, show the similar speed of divergence toward $\ka _i$ as the beta density does. These observations imply that, with the appropriate choice of hyperparaemters, the log-adjusted shrinkage prior can introduce the strong tail-robustness without losing the horseshoe-type shrinkage effect. 

\begin{figure}[!htbp]
	\centering
	\includegraphics[width=13cm]{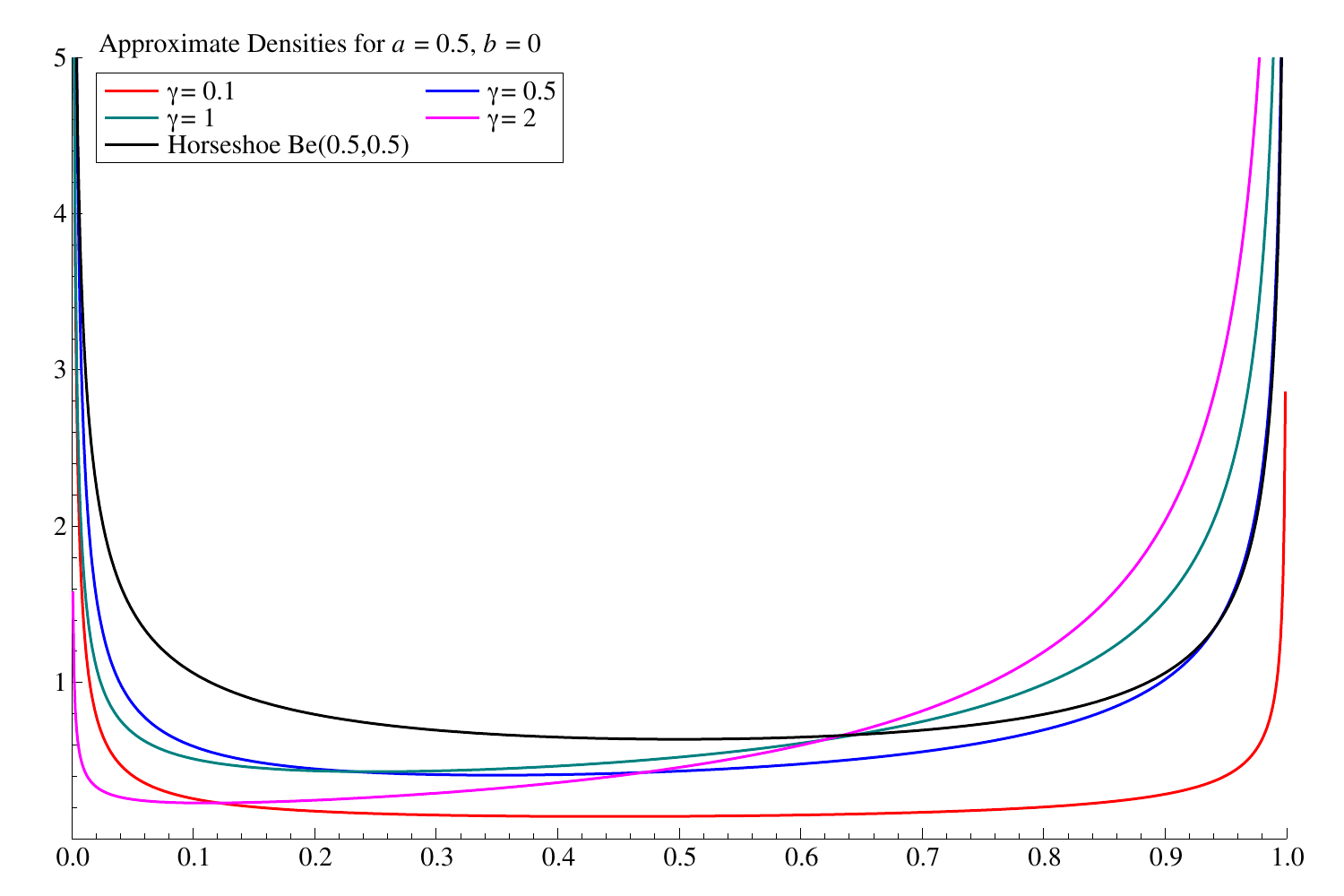} 
	\caption{
		The prior density for the shrinkage factor $\kappa_i$ given in (\ref{EHS}) under the log-adjusted shrinkage prior with $a=1/2$ and $b=0$. The four densities represent the cases of $\ga = 0.1$ (red), $0.5$ (blue), $1$ (green) and $2$ (pink). 
		The density of beta distribution $Be(1/2,1/2)$ (black) is equivalent to the half-Cauchy prior in the scale of $u_i$ and realizes the horseshoe prior. The log-adjusted shrinkage priors have the steep increase toward $\ka _i=0$ for tail-robustness, while maintaining its spike around $\ka _i = 1$ for strong shrinkage if we choose moderate values of $\ga$, such as $\ga = 0.5$ or $\ga =1$. 
		\label{fig:dens}
	}
\end{figure}

In order to clarify the importance of setting $b=0$, we next examine the posterior tail-robustness under the proposed prior by computing the posterior mean squared error (MSE) for large $y_i$. Denote the posterior MSE under a prior $\pi ( \theta_i )$ by ${\rm{MSE}}_{\pi } ( \th _i | y_i ) = E_{\pi } [ ( \th _i - y_i )^2 | y_i ]$. We evaluate the MSE of the proposed class of priors in the following theorem based on the representation of posterior MSE by the marginal likelihood (e.g., see \citealt{polson1991representation}).

\begin{thm}
\label{thm:MSE1}
Under the log-adjusted shrinkage priors for $\theta_i$ with hyperparameters $a>0, b\geq 0$ and $\gamma>0$, and under the beta-type prior ($a>0$, $b>0$ and $\ga = -1$), it holds that 
\begin{align}\label{MSE} 
{\rm{MSE}}_{\pi } ( \th _i | y_i )=1+\frac{2}{y_i^2}(1 + b) (1 + 2 b)+o\left(\frac{1}{y_i^2}\right), 
\end{align}
where ${y_i}^2 o(1 / {y_i}^2 ) \to 0$ as $|y_i|\to\infty$. 
\end{thm}

\noindent 
We first note that the above approximation formula of MSE is independent of $\ga$, thereby the same formula holds for the original beta-type shrinkage prior for $\theta_i$.
Moreover, Theorem \ref{thm:MSE1} shows that the leading term of the posterior MSE for large signals is increasing in $b$, which clearly suggests that the best choice is $b=0$. 
The proper log-adjusted shrinkage prior can attain the ideal MSE by setting $b=0$, outperforming in the MSE the proper beta-type shrinkage prior for which we always have to set $b>0$.

%
The posterior MSE for large $y_i$ has also been calculated for other shrinkage priors. Theorem 7 of \cite{bhadra2017horseshoe+} provided the posterior MSE of the horseshoe$+$ prior $\pi _{{\rm{HS}}+}$ as
\begin{equation*}
{\rm{MSE}}_{\pi _{{\rm{HS}}+}} ( \th _i | y_i ) = 1 + 3 \left( \frac{2}{ {y_i}^2 } \right) + o\left( \frac{1}{ {y_i}^2 } \right), 
\end{equation*} 
which is, as $| y_i | \to \infty $, inevitably larger than the posterior MSE of our proposed prior in (\ref{MSE}) with $b = 0$. 

\subsection{Posterior computation}
Although the Bayes estimator of $\theta _i$ for the log-adjusted shrinkage prior is not analytically available, there is an efficient yet simple Markov chain Monte Carlo algorithm for posterior computation. 
The full conditional posteriors of parameters, $\theta _i$'s, $u_i$'s and $\tau$, become well-known distributions after the appropriate augmentation by latent variables described below. The conditional posterior density of hyperparameter $\ga$ is complex due to the intractable normalizing constant $C(a,b,\ga )$, but the sampling from its distribution is feasible by the accept-reject algorithm. The rest hyperparameters are fixed as $a=1/n$ and $b=0$. 

The prior density of $u_i$ in (\ref{EHS-u}) has the following augmented expression: 
\begin{equation}\label{Aug}
\pi(u_i; \gamma)
= \frac{1}{C(a,b,\ga )}
\int _0^{\infty}\!\!\!\int _0^{\infty} u_i^{a-1}\frac{v_i^{\gamma}e^{-v_i}}{\Gamma(1+\gamma)}\frac{w_i^{v_i+a+b-1}e^{-w_i(1+u_i)}}{\Gamma(v_i+a+b)}dw_idv_i, 
\end{equation}
where $w_i$ and $v_i$ are latent variables for data augmentation. 
Given $\ga$, the augmented posterior distribution is proportional to 
$$
\pi(\tau)\prod_{i=1}^n {\rm{N}} ( y_i | \th _i , 1) {\rm{N}} ( \th _i | 0, \ta u_i ) \frac{1}{C(a,b,\ga )} u_i^{a-1}v_i^{\gamma}e^{-v_i}\frac{w_i^{v_i+a+b-1}e^{-w_i(1+u_i)}}{\Gamma(v_i+a+b)},
$$
where $\pi(\tau)$ denotes a prior distribution of $\tau$.
We assign an inverse-gamma prior for $\tau$ and set $\pi ( \ta ) = IG(\ta | c_{0}^{\ta } , d_{0}^{\ta } )$, which leads to conditional conjugacy. 
It is immediate from the expression above that all the full conditional distributions are of normal or (inverse) gamma distributions, so that we can efficiently carry out Gibbs sampling by generating posterior samples from those distributions. The procedure of Gibbs sampler is summarized as follows:

\begin{algo}[Gibbs sampling algorithm]\label{algo:gibbs}
Suppose $\ga$ is fixed. The Gibbs sampler algorithm, or the list of the full conditional distributions of the local and global parameters under the log-adjusted shrinkage prior, is summarized as follows: 
\begin{itemize}
\item
Generate $\ta $ from ${\rm{IG}} (n / 2 + c_{0}^{\ta } , \sum_{i = 1}^{n} {\th _i}^2 / (2 u_i ) + d_{0}^{\ta } )$. 

\item
Generate $\th_i$ from ${\rm{N}} ( y_i / \{ 1 + 1 / ( \ta u_i ) \} , 1 / \{ 1 + 1 / ( \ta u_i ) \} )$ for $i = 1, \dots , n$. 

\item
Generate $u_i$ from the generalized inverse Gaussian full conditional distribution $p( u_i | \th _i , w_i , \ta ) \propto {u_i}^{- 1 / 2 + a - 1} \exp [- \{ 2 w_i u_i + ( {\th _i}^2 / \ta ) / u_i \} / 2]$ for $i = 1, \dots , n$. 

\item Generate $(v_i,w_i)$ by the following two steps: 

\begin{itemize}
\item
Generate $v_i$ from the conditional distribution marginalized over $w_i$, namely $p(v_i|\theta _i, \tau , y_i)$, which is ${\rm Ga}(1+\gamma,1+\log(1+u_i))$, for $i=1,\ldots, n$. 

\item
Generate $w_i$ from the full conditional distribution, $p(w_i|v_i,\theta _i, \tau , y_i)$, which is ${\rm Ga}(v_i+a+b,1+u_i)$ for $i=1,\ldots,n$.	
\end{itemize}

\end{itemize}
\end{algo}

%
We next consider incorporating the estimation of $\gamma$ into the algorithm above. Let $C(\gamma)= C(a,b,\ga )$ be the normalizing constant of $\pi (u_i)$ in (\ref{EHS-u}), which is defined by 
\begin{equation}\label{NC}
C(\gamma) = \int_{0}^{\infty } {u^{a - 1} (1 + u)^{-(a + b)} \over \{ 1 + \log (1 + u) \} ^{1 + \ga }} du = \int_0^1 {(1-\kappa)^{a-1}\kappa^{b-1} \over (1-\log\kappa)^{1+\gamma }} d\kappa \text{,} 
\end{equation}
where $\kappa = 1/(1+u)$, $a=1/n$ and $b=0$. 
Due to this intractable normalizing constant in the prior, the direct sampling from the full conditional of $\gamma$ is challenging and even infeasible. We circumvent this problem by constructing upper and lower bounds for the normalizing constant, which allows for the independent Metropolis-Hastings algorithm by providing the bounds of the acceptance probability with arbitrary accuracy. This approach is similar to the alternating series method \citep{devroye1981series,devroye2009exact} and widely used in, for example, the sampling from the Polya-gamma distribution \citep{polson2013bayesian}. 
Our sampling procedures are briefly sketched in the following. 
Details of the bounds of the acceptance probability, $\overline{w}$ and $\underline{w}$, and some remarks about the sampling procedures are provided in Appendix \ref{app:gamma}. 

We use the gamma prior $Ga(a_0^{\gamma },b_0^{\gamma })$ for $\ga $. The proposal distribution is $Ga(a_1^{\gamma},b_1^{\gamma})$, whose parameters are given by
\begin{equation*}
 a_1^{\gamma} = a_0^{\gamma} + n a \ \ \ \mathrm{and} \ \ \ b_1^{\gamma} = b_0^{\gamma} + \sum _{i=1}^n \log (1+\log (1+u_i) ). 
\end{equation*}
Denote the current state at an iteration of the MCMC algorithm by $\gamma$, and the candidate drawn from the proposal by $\gamma '$. 
The acceptance probability $A(\gamma \to \gamma ') $ is bounded below and above by $\underline{w}$ and $\overline{w}$. 
Both are functions of $(a,\gamma , \gamma ' , K)$ and converge to $A(\gamma \to \gamma ')$ as $K\to \infty$, where $K$ controls the precision of the approximation and can be set as large as necessary. 
The procedure of the independent Metropolis-Hastings sampling is summarized as follows. 

\begin{algo}[Sampling from the full conditional of $\gamma$] 
\label{algo:gamma}
The steps for generating $\gamma$ from its full conditional distribution is summarized as follows: given the current sample $\gamma$,

\begin{enumerate}
	\item Generate $\gamma '$ from the proposal $Ga(a_1^{\gamma},b_1^{\gamma})$. 
		
	\item Generate $U \sim U(0,1)$.
	
	\item Given $K$, evaluate $\underline{w}$ and $\overline{w}$. Then, 
	
	\begin{itemize}
	\item If $U<\underline{w}$, accept $\gamma '$ as the sample of this iteration. 
	
	\item If $U>\overline{w}$, reject $\gamma '$ and keep $\gamma$ as the sample of this iteration. 
	
	\item Otherwise ($\underline{w} < U < \overline{w}$), increase $K$ and redo step 3. 
	\end{itemize}

\end{enumerate}
\end{algo}


\subsection{Generalization using iterated logarithm} 
\label{sec:extensions} 

Following the motivation given in the introduction, the log-adjusted shrinkage prior is further extended to the more general class of distributions. As the (scaled) beta distributions is extended to the log-adjusted version by the multiplicative log-term, this generalization is naturally realized by the use of finitely iterated logarithmic functions. 

For $z \ge 1$, let $f_1(z) \equiv f(z) \equiv 1+\log (z)$. 
Then, the iterated logarithm is defined inductively by $f_{L+1}(z) \equiv f( f_{L}(z) )$ for $L = 1, 2, \dotsc $. 
Define the extended version of the modified scaled beta priors with parameter $\ga > 0$ by 
\begin{equation} \label{ILAS}
\pi(u_i; \ga, L ) \propto  \frac{u_i^{a - 1}}{(1 + u_i)^{a+b}} \left\{ \prod _{k=1}^{L-1} \frac{1}{f_k (1+u_i)} \right\} \frac{1}{f_L (1+u_i)^{1 + \ga }} \text{,} 
\end{equation}
where $a>0$ and $b\ge 0$ are constant (We again recommend $a=1/n$ and $b=0$). The corresponding prior for the shrinkage factor $\kappa_i=(1+u_i)^{-1}$ is given by 
\begin{equation*} 
\pi(\kappa_i; \ga, L ) \propto \kappa_i^{b-1}(1-\kappa_i)^{a-1} \left\{ \prod _{k=1}^{L-1} \frac{1}{f_k (1/\kappa_i)} \right\} \frac{1}{f_L (1/\kappa_i)^{1 + \ga }} \text{.} 
\end{equation*}
The prior for $\th _i$ induced by this distribution as the scale mixture of normal is named {\it iteratively log-adjusted shrinkage (ILAS) prior}. 

When $L = 1$, this prior is precisely the original log-adjusted shrinkage prior discussed in the previous subsections. We require that $b=0$ for the improvement from the beta-type prior, but the priors are still proper if only $\ga >0$, as shown in the following theorem. 
As order $L$ increases, the tails of the density for $\th_i$ becomes heavier, while remaining in the class of proper priors, from which we expect the stronger tail-robustness of the Bayes estimators. 

\begin{thm}
\label{thm:properties_of_IL} 
The following properties hold under the iterative log-adjustment. 
\begin{itemize}
\item[1.]
The iteratively log-adjusted shrinkage prior for $\theta_i$ with finite $L$ holds the same properties given in Theorem \ref{thm:prop}. 

\item[2.]
Suppose that $b = 0$. 
Then, for any $0 < \ep < 1$, the prior probability of the log-adjusted beta distribution for $\ka _i$ falling in the interval $(0, \ep )$ tends to $1$ as $L \to \infty $; namely  
\begin{align}
\lim_{L \to \infty } \int_{0}^{\ep } \pi( \ka _i ; \ga, L ) d{\ka _i} = 1 \text{.} \non 
\end{align}
\end{itemize}
\end{thm}

\noindent 
The first property indicates that the iterative log-adjustments do not change the original properties of the proposed prior, including integrability, density spike around the origin and heavier tails. 
The second statement shows the convergence of the iterated log-adjusted shrinkage prior to the point mass on $\kappa _i = 0$ in distribution as $L\to \infty$. In the limit, the proposed prior does not shrink the outliers at all. 
However, losing the shrinkage effect at all is not desirable, and we fix $L$ at some finite value so that the prior density keeps the steep spike around zero. 

Although it is difficult to draw the density functions of the iteratively log-adjusted shrinkage priors as in Figure~\ref{fig:dens} for the intractable normalizing constant, the newly-multiplied log-terms can easily be evaluated and shown in Figure~\ref{fig:log}. 
It is clear in the top figure that function $f_L(1+u)$ is increasing in $u$, but converges to the constant function as $L\to \infty$, which are also verified in Appendix \ref{app:IL-prop}.
This observation implies that the marginal effect of log-terms being multiplied to the prior is diminishing as $L$ increases. The bottom figures displays the reciprocal of the iterative log-terms in the scale of $\ka _i$, which are actually multiplied to the original log-adjusted shrinkage prior. The lower densities near $\ka _i = 0$ moderates the spike and makes the density integrable, while the iterative log-term is unity around $\ka _i=1$ and affects the shrinkage effect less. 

\begin{figure}[!htbp]
	\centering
	\includegraphics[width=13cm]{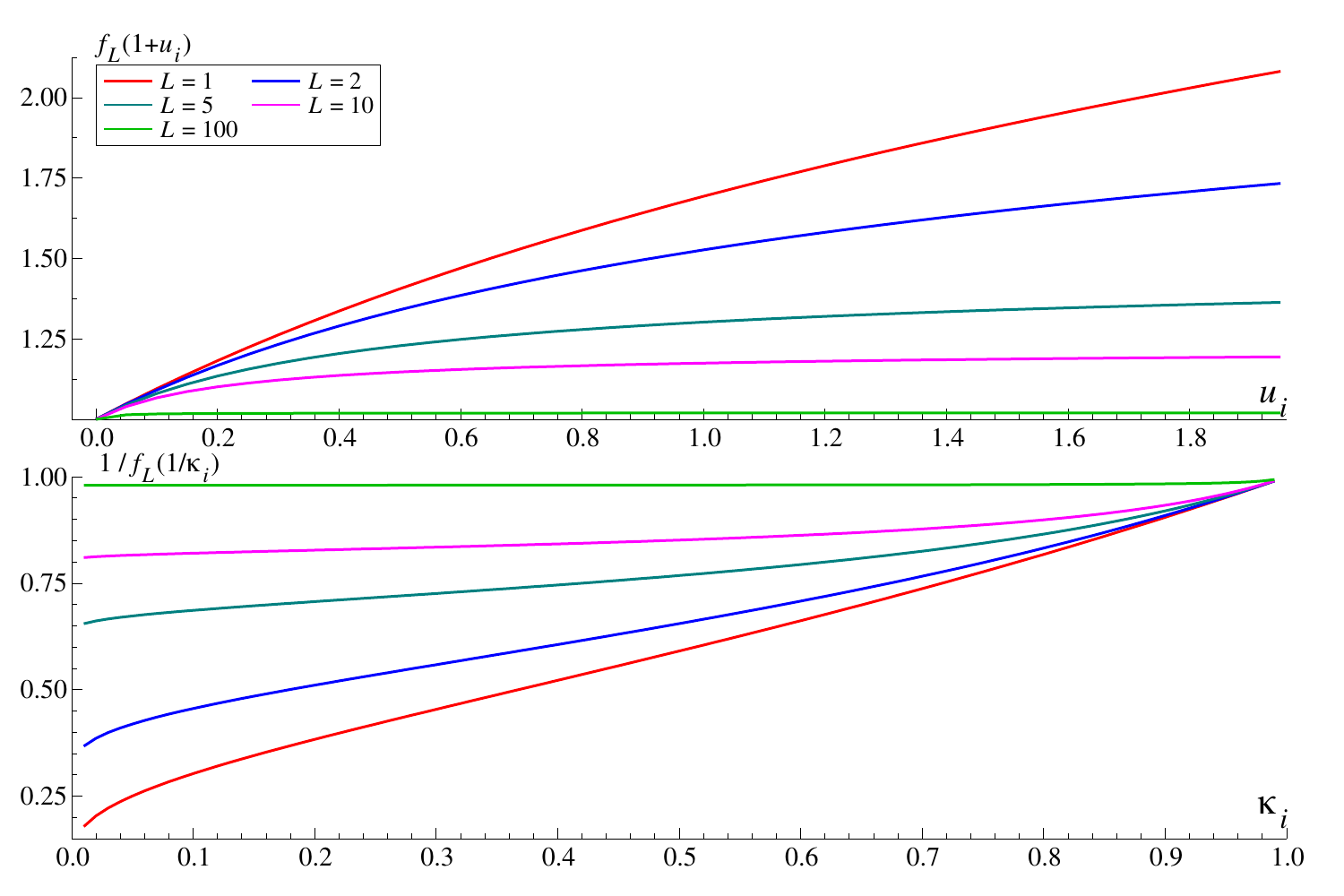} 
	\caption{
		The functions $f_L(1+u_i)$ (top) and $1 / f_L(1/\ka _i)$ (bottom) with $L=1$ (red), $2$ (blue),  $5$ (green), $10$ (pink) and $100$ (light green). It is evident that the repeated application of operation $f$ makes the function closer to constant. In the bottom figure, the decrease of functions as $\ka _i \to 0$ moderates the divergence of the prior density around zero and contributes to the integrability of the iteratively log-adjusted shrinkage priors. 
		\label{fig:log}
	}
\end{figure}

The posterior mean squared error under the iteratively log-adjusted shrinkage prior can also be computed in the similar way as in the proof of Theorem~\ref{thm:MSE1}. 

\begin{thm}
	\label{thm:MSE2}
	The posterior mean squared error under the iteratively log-adjusted shrinkage prior 
	satisfies, for $b=0$,
	\begin{align}
	{\rm{MSE}}_{\pi_{IL}} ( \th _i | y_i ) 
	&= 1 + {1 \over {y_i}^2 / 2} \Big\{ 1 - {3 a / 2 \over {y_i}^2 / 2} +\sum_{k = 1}^{L - 1} {1 \over f_k (1 + {y_i}^2 / 2)} \dotsm {1 \over f_1 (1 + {y_i}^2 / 2)} \non \\
	&\quad + (1 + \ga ) {1 \over f_L (1 + {y_i}^2 / 2)} \dotsm {1 \over f_1 (1 + {y_i}^2 / 2)} \Big\} + o \Big( {1 \over {y_i}^2} \Big) \text{,} \label{eq:MSE-IL} 
	\end{align}
	where ${y_i}^2 o(1 / {y_i}^2 ) \to 0$ as $|y_i|\to\infty$. 
\end{thm}

\noindent 
Theorem~\ref{thm:MSE2} derives the higher-order terms of $y_i$ that is ignored in the MSE under the original log-adjusted shrinkage prior in Theorem~\ref{thm:MSE1}, whose leading term is simply $1 + 2 /y_i^2$.
We summarize our findings on the derived MSE in the following three points. First, this result reveals that the effect of hyperparameters $a$ and $\ga$ is limited to the higher-order terms, which is consistent with the result of Theorem~\ref{thm:MSE1}. In addition, the choice of hyperparameter $\ga$ is less sensitive to the Bayes estimator if $L$ is large. Secondly, there is no difference in the MSEs under the original and iteratively log-adjusted shrinkage priors in the order of $2/y_i^2$, while both priors are still superior to the beta-type shrinkage priors in the MSE in the tail. Finally, increasing the order $L$ has no negative effect on the point estimation so long as the posterior MSE under large signals is concerned. In fact, it is difficult to understand whether one increment of $L$ increases or decreases the MSE in this expression; it is determined together with the values of $y_i$ and $\ga$. We revisit this issue partially in the simulation studies in Section~\ref{sec:numerical_study}. 

The posterior computation with the iteratively log-adjusted shrinkage prior is also straightforward. 
The parameter augmentation given in (\ref{Aug}) can be generalized for the new density of $u_i$ in (\ref{ILAS}) as, ignoring the normalizing constant,
\begin{align}
\pi (u_i; \ga , L) \propto  \int_{(0, \infty )^{L + 1}} u_i^{a - 1}e^{- t_{iL} u_i} {\rm{GS}}_L ( t_{i,0:L} | \ga ) d{t_{i,0:L}} \text{,} \label{eq:augmentation_IL} 
\end{align}
where $t_{i,0:L}=(t_{i0},t_{i1},\ldots,t_{iL})$ and ${\rm GS}_L ( t_{i,0:L} | \ga )$ is the joint density of a non-stationary Markov process defined by 
\begin{align}
{\rm{GS}}_L ( t_{i,0:L} | \ga ) &= {\rm{Ga}} ( t_{i0} | 1 + \ga , 1) {\rm{Ga}} ( t_{i1} | t_{i0} + 1, 1) \times \dotsm \non \\
&\quad \times {\rm{Ga}} ( t_{i, L - 1} | t_{i, L - 2} + 1, 1) {\rm{Ga}} ( t_{iL} | t_{i,L - 1} + a + b, 1) \text{.} \non 
\end{align}
The density of $u_i$ is the shape mixture of density kernel of gamma distribution by a gamma-shape Markov process. 
The integral expression above defines latent variables $t_{il}$'s and gives the following tractable full conditional distributions for Gibbs sampler.

\begin{algo}[Gibbs sampler for local parameters under ILAS prior]
\label{algo:gibbs-ILAS}

The sampling steps for local parameters $u_i$ and $t_{0:L}$ are summarized as follows: 
\begin{itemize}
\item 
The full conditional distribution of $u_i$ is ${\rm{GIG}}(-1/2 + a, 2t_{iL}, \theta_i^2/\tau )$. 
		
\item
The full conditional distribution of $t_{0:L}$ has the compositional form, 
\begin{align}
&{\rm{Ga}} ( t_{i, 0} | 1 + \ga , f_L (1 + u_i)) {\rm Ga}( t_{i, 1} | t_{i, 0} + 1, f_{L - 1} (1 + u_i)) \times \dotsm \non \\
&\quad \times {\rm Ga}( t_{i,L - 1} | t_{i,L - 2} + 1, f_{1} (1 + u_i)) {\rm{Ga}} ( t_{i, L} | t_{i,L - 1} + a + b, 1 + u_i) \text{,} \non 
\end{align}
thereby the random samples can be sequentially generated. 
\end{itemize}
\end{algo}

\noindent 
The above procedure can be incorporated into Algorithm \ref{algo:gibbs}, which enables us to efficiently generate posterior samples of $\theta_i$. 
It is worth noting that $t_{i, k}$, $k \le L - 1$, are not used to generate samples of $\{ u_i , \th _i , \ta \} $.

The shrinkage priors with logarithm terms in their densities have been studied in various ways. We considered the prior distributions proposed in \cite{bhadra2017horseshoe+} and \cite{womack2019heavy} and confirmed that their prior densities could be extended in a similar way by repeatedly multiplying the additional terms to the density function. However, such iterative operation is extremely complex, compared with the simple recursive construction of the additional terms in this research, $f_{L+1}(z) = 1 + \log f_L(z)$, that defines the log-adjusted shrinkage priors.

\section{Numerical Study}
\label{sec:numerical_study} 

\subsection{Simulation study}
We illustrate finite-sample performance of the Bayes estimators under the proposed priors and other shrinkage priors proposed in the recent research in various situations of true sparse signals. 
We generated $n = 200$ observations from $y_i\sim N(\theta_i,1)$, where $\th_i$ is a true signal. 
We adopted the following two scenarios for $\theta_i$:
\begin{align*}
&{\rm (I)} \ \  \theta_i\sim \frac{\omega}{2}\delta(c)+\frac{\omega}{2}\delta\left(-\frac{c}2\right)+\left(1-\omega\right)\delta(0),\\ 
&{\rm (II)} \ \ \theta_i\sim \frac{\omega}{2}N(c,1)+\frac{\omega}{2}N\left(-\frac{c}{2},1\right)+\left(1-\omega\right)\delta(0),
\end{align*}
where $\delta(x)$ denotes the one-point distribution on $x$. Weight $\omega$ controls the sparsity level in the signals $\theta_i$; smaller value of $\omega$ leads to more sparsity. $c$ is the locations of non-null signals.
We considered six settings of $\omega$ and $c$ as the combinations of $\omega\in \{0.1,0.2,0.3\}$ and $c\in\{6,9\}$.

For the simulated dataset, we applied three types of proposed priors: the log-adjusted shrinkage prior with $a=1/n, b=0$ and $\ga=1$ (denoted by LAS), an adaptive version of the LAS prior with a fully Bayesian approach for $\gamma$ (denoted by aLAS), and the iteratively log-adjusted shrinkage prior with $a=1/n, b=0, \gamma=1$ and $L=3$, denoted by ILAS. 
As competitors, we also applied the Horseshoe prior \citep[HS;][]{carvalho2010horseshoe}, normal-beta prime prior \citep[NBP;][]{bai2019large}, Dirichlet-Laplace prior \citep[DL;][]{Bhatta2015}, and Horseshoe+ prior \citep[HS+;][]{bhadra2017horseshoe+}. 
To implement the posterior analysis with the HS+ and DL priors, we employed R package ``NormalBetaPrime'' \citep{bai2019large,bai2020beta} with default settings, such as the use of uniform prior on $(1/n,1)$ for the global scale parameter. 
For the other models, we adopted $\tau \sim C^{+}(0,1/n)$. 
In applying all the priors, we generated 1000 posterior samples after discarding the first 1000 samples as burn-in period, and computed posterior means of $\th_i$. 
The squared error losses of the posterior means $\thh_i$, give by $\sum_{i=1}^n(\thh_i-\th_i)^2$, were calculated and averaged over 500 replications of simulations. 

The results are reported in Table \ref{tab:sim}. 
It shows that the methods are comparable when $\omega$ is small and the true signals are very sparse. 
On the other hand, as $\omega$ increases, the proposed priors get appealing compared with the other methods. 
This result is consistent with Theorems \ref{thm:MSE1} and \ref{thm:MSE2} and reflects the fact that the proposed priors have heavier tails to accommodate large signals. 
It is also observed that the proposed three methods are almost equally successful and it is difficult to discuss their superiority. 
Focusing the comparison on the performance of LAS and aLAS, the benefit from estimating the adjustment parameter $\gamma$ could be limited possibly because of the trade-off between flexibility of data-adaptive selection of $\gamma$ and inflation of uncertainty arising from estimating $\gamma$. 
Finally, LAS and ILAS perform quite similarly in every setting, which could be related to the fact that these two priors differ only in the form of the slowly varying part. It can also be explained by their MSEs in the tails that are exactly the same up to the order of $1 / {y_i}^2$.

\begin{table}[!htb]
\caption{Comparison of averaged values of squared error losses of the posterior mean estimates of $\theta_i$ under the fixed log-adjusted shrinkage (LAS) prior, the adaptive LAS prior (aLAS), the iteratively log-adjusted shrinkage prior (ILAS) of order three (EH-IL), the Horseshoe prior (HS), the normal-beta prime prior (NBP), the Dirichlet-Laplace prior (DL), and the Horseshoe+ prior (HS+). 
The lowest averaged squared error loss for each setting (in rows) is in bold.
\label{tab:sim}
}
\begin{center}
\begin{tabular}{ccccccccccc}
\hline
 c & omega &  & LAS & aLAS & ILAS & HS & NBP & DL & HS+ \\
\hline
  & 0.1 &  & {\bf 52.5} & 55.0 & 52.5 & 55.9 & 55.7 & 62.7 & 54.2 \\
 6 & 0.2 &  & 90.6 & {\bf 89.9} & 90.7 & 96.6 & 107.6 & 120.5 & 102.7 \\
  & 0.3 &  & 124.3 & {\bf 121.3} & 124.2 & 126.4 & 159.5 & 177.9 & 150.5 \\
 \hline
   & 0.1 &  & 43.3 & 47.9 & 43.0 & 49.4 & 41.7 & 50.9 & {\bf 41.4} \\
 9 & 0.2 &  & {\bf 68.6} & 71.5 & {\bf 68.6} & 87.7 & 76.2 & 95.5 & 73.7 \\
   & 0.3 &  & {\bf 92.5} & 94.3 & {\bf 92.5} & 117.8 & 110.4 & 137.5 & 105.3 \\
 \hline
   & 0.1 &  & 47.9 & 51.1 & {\bf 47.6} & 50.9 & 49.3 & 56.2 & 48.2 \\
 6 & 0.2 &  & {\bf 84.3} & 84.7 & 84.4 & 92.2 & 98.4 & 112.4 & 94.1 \\
   & 0.3 &  & 113.5 & {\bf 111.8} & 113.7 & 120.5 & 143.1 & 164.8 & 135.7 \\
 \hline
   & 0.1 &  & 43.9 & 48.6 & 43.7 & 49.2 & 41.7 & 49.2 & {\bf 41.6} \\
 9 & 0.2 &  & {\bf 71.6} & 74.5 & {\bf 71.6} & 88.6 & 78.3 & 93.6 & 76.0 \\
  & 0.3 &  & {\bf 96.0} & 97.7 & 96.3 & 119.0 & 113.2 & 136.2 & 108.6 \\
\hline
\end{tabular}
\end{center}
\end{table}

\subsection{Example: Prostate Cancer Data}
We demonstrate real-data application of the proposed priors using a popular prostate cancer dataset in \cite{Singh2002}. 
In this dataset, there are gene expression values for $n=6033$ genes for $m = 102$ subjects, with $m_1 = 50$ normal control subjects and $m_2 = 52$ prostate cancer patients. 
The goal of this analysis is to identify genes that are significantly different between the two groups. 
We first conduct $t$-test for each gene to compute the test statistics $t_1,\ldots,t_n$, and then transform them to $z$-scores through $z_i=\Phi^{-1}(F_{m-2}^t(t_i))$, where $\Phi(\cdot)$ is the standard normal distribution function and $F_k^t(\cdot)$ is the distribution function of $t$-distribution with $k$ degrees of freedom.
For the resulting $z$-scores, $z_1,\ldots, z_n$, we applied the following model: 
$$
z_i=\theta_i+\ep_i, \ \ \ \ep_i\sim N(0,1), \ \ \ i=1,\ldots,n.
$$
We again compare the same seven priors for $\th _i$ as in the previous subsection. 
Based on 5000 posterior samples after discarding the first 5000 samples, we computed posterior means of $\th_i$.
In Table \ref{tab:exm}, we presented top 10 genes selected by \cite{Efron2010} and their estimated effect size $\th_i$ on prostate cancer.
The absolute value of each effective size estimate is largest for aLAS, because of its strong tail-robustness. 
However, the estimates of all the seven methods do not differ drastically from one another.

\begin{table}[!htb]
\caption{The $z$-scores and the effect size estimates based on posterior means for the top 10 genes selected by Efron under the seven priors. 
\label{tab:exm}
}
\begin{center}
\begin{tabular}{crrrrrrrrrrrr}
\hline
Gene & $z$-score & LAS & aLAS & ILAS & HS & NBP & DL & HS+ \\
\hline
610 & 5.29 & 4.87 & 5.00 & 4.58 & 4.89 & 4.90 & 4.58 & 4.88 \\
1720 & 4.83 & 4.29 & 4.52 & 4.11 & 4.33 & 4.36 & 4.09 & 4.40 \\
332 & 4.47 & 3.98 & 4.13 & 3.68 & 3.86 & 3.88 & 3.62 & 3.89 \\
364 & -4.42 & -3.85 & -4.09 & -3.59 & -3.80 & -3.90 & -3.64 & -3.76 \\
914 & 4.40 & 3.78 & 3.98 & 3.63 & 3.79 & 3.85 & 3.58 & 3.30 \\
3940 & -4.33 & -3.78 & -4.00 & -3.46 & -3.71 & -3.76 & -3.50 & -3.61 \\
4546 & -4.29 & -3.70 & -3.91 & -3.37 & -3.40 & -3.64 & -3.36 & -3.56 \\
1068 & 4.25 & 3.61 & 3.86 & 3.35 & 3.45 & 3.67 & 3.34 & 3.66 \\
579 & 4.19 & 3.61 & 3.81 & 3.33 & 3.33 & 3.54 & 3.27 & 3.55 \\
4331 & -4.14 & -3.61 & -3.80 & -3.28 & -3.27 & -3.54 & -3.14 & -3.42 \\
\hline
\end{tabular}
\end{center}
\end{table}

\section{Discussion}
\label{sec:discussion} 

In this research, the repeated multiplication of the log-terms to the density is proven successful in defining the new class of distributions that are continuous, proper and extremely heavy-tailed. Although the focus of this research is on tail-robustness, it is also natural to consider the idea of log-adjustment to define the stronger shrinkage effect. To be precise, {\it the doubly log-adjusted shrinkage prior}, whose density in the scale of $\ka _i$ is given by, 
\begin{align}
\pi (\ka _i;\alpha , \beta,L) &\propto 
{\ka _i}^{- 1} (1 - \ka _i )^{- 1} \non \\
& \times \left\{ \prod _{l=1}^{L-1} f_l\left(\frac{1}{\ka _i}\right)^{-1} f_l\left(\frac{1}{1-\ka _i}\right)^{-1} \right\} f_{L}\left(\frac{1}{\ka _i}\right) ^{-(1+\al)} f_{L}\left(\frac{1}{1-\ka _i}\right) ^{-(1+\be)} \non 
\end{align}
is of great interest. We proved that, as iteration $L$ increases, if we choose the sequence of hyperparameters $(\al _L,\be _L)$ appropriately, then the prior, $\pi (\ka _i;\alpha _L , \beta _L, L)$, converges in distribution to the point masses on $\{ \ka _i = 0 \}$ and $\{ \ka _i = 1 \}$, i.e., the spike-and-slab prior. For the details of the proof, see Appendix~\ref{app:double-IL}. The resemblance to the degenerate prior shown in this result could justify the use of iteratively log-adjusted priors as the continuous alternative of the degenerate variable-selection priors. Although the finite-sample properties of Bayes estimators under the prior above is not developed here, the posterior inference with this prior is feasible by the same augmentation we proved for the iteratively log-adjusted priors. We believe that the priors with iterated logarithm is the promising future research in exploring the class of shrinkage priors with logarithms.

\vspace{1cm}
\bibliographystyle{chicago}
\bibliography{MTPB}

\newpage
\appendix 

\setcounter{equation}{0}
\renewcommand{\theequation}{A\arabic{equation}}
\setcounter{section}{0}
\renewcommand{\thesection}{A\arabic{section}}
\setcounter{lem}{0}
\renewcommand{\thelem}{A\arabic{lem}}

\begin{center}
{\bf\Large Appendix}
\end{center}


\section{Proof of Theorem \ref{thm:prop}}
\label{app:proof_of_prop} 
This proof can be obtained as the special case of the proof of Theorem \ref{thm:properties_of_IL} with $L=1$ given in Appendix \ref{app:IL-proof}.

\section{Proof of Theorem \ref{thm:MSE1}}
\label{app:proof_of_MSE1} 

%

%

%

We first provide a useful lemma.
For details, see, for example, the discussion at the end of Section 1.2 of \cite{senetaregularly}. 

\begin{lem}
\label{lem:SV}
Let $L(u)$ be a strictly positive and continuously differentiable function of $u > 0$. 
Suppose that 
\begin{align}
\lim_{u \to \infty } {u L' (u) \over L(u)} = 0 \text{.} \non 
\end{align}
Then the function $L(u)$ is slowly varying as $u \to \infty $, that is, $\lim_{M\to\infty}L(Mv)/L(M)=1$ for all $v>0$. 
\end{lem}

We will suppress the subscript $i$ and write $u$, $\th $, and $y$ for $u_i$, $\th _i$, and $y_i$, respectively, for notational simplicity. 
Let $p( \th )$ and $m(y)$ denote the marginal densities of $\th $ and $y$ under the LAS prior $\pi (u) \propto u^{a - 1} (1 + u)^{- a - b} \{ 1 + \log (1 + u) \} ^{- (1 + \ga )}$. 
We define $S(u)$ as 
\begin{equation}
\label{SVF}
S(u) = \Big( {u \over 1 + u} \Big) ^{a + b} \big\{ 1 + \log (1 + u) \big\} ^{-(1 + \ga) },
\end{equation}
so that $\pi(u)=C^{-1}u^{-b-1}S(u)$ with normalizing constant $C = \int_{0}^{\infty } u^{-b-1} S(u) du$. 
From Lemma \ref{lem:SV}, it can be shown that $S(\cdot)$ is a slowly varying function.

We first note that the posterior MSE can be written as
\begin{align}
{\rm{MSE}}_{\pi } ( \th | y) &= 1 + {1 \over m(y)} {\partial ^2 m(y) \over {\partial y}^2}, \non 
\end{align} 
since the second order derivative of $m(y)$ can be expressed as  
\begin{align}
{\partial ^2 m(y) \over {\partial y}^2} &= \int_{- \infty }^{\infty } \Big[ {\partial ^2 \over {\partial y}^2} {1 \over \sqrt{2 \pi}} \exp \Big\{ - {(y - \th )^2 \over 2} \Big\} \Big] p( \th ) d\th \non \\
&= \int_{- \infty }^{\infty } \{ - 1 + ( \th - y)^2 \} {1 \over \sqrt{2 \pi}} \exp \Big\{ - {(y - \th )^2 \over 2} \Big\} p( \th ) d\th 
\text{.} \non 
\end{align}
On the other hand, since $y | u \sim {\rm{N}} (0, 1 + u)$, we have that
\begin{align}
m(y) &= {1 \over \sqrt{2 \pi }} \int_{0}^{\infty } {1 \over \sqrt{1 + u}} \exp \Big( - {y^2 / 2 \over 1 + u} \Big) \pi (u) du \non 
\end{align}
and hence that 
\begin{align}
{\partial ^2 m(y) \over {\partial y}^2} &= - {1 \over \sqrt{2 \pi }} \int_{0}^{\infty } {1 \over (1 + u)^{3 / 2}} \exp \Big( - {y^2 / 2 \over 1 + u} \Big) \pi (u) du \non \\
&\quad + {y^2 \over \sqrt{2 \pi }} \int_{0}^{\infty } {1 \over (1 + u)^{5 / 2}} \exp \Big( - {y^2 / 2 \over 1 + u} \Big) \pi (u) du \text{.} \non 
\end{align}
Therefore, by making the change of variables $u = ( y^2 / 2) v$, it follows that
\begin{equation}\label{post-MSE}
{1 \over m(y)} {\partial ^2 m(y) \over {\partial y}^2}
=
\frac{2}{y^2}\frac{2I(y,5/2)-I(y,3/2)}{I(y,1/2)},
\end{equation}
where
$$
I(y,k) = \int_{0}^{\infty } \Big\{ {y^2 / 2 \over 1 + ( y^2 / 2) v} \Big\} ^{ k } v^{- b - 1} \exp \Big\{ - {y^2 / 2 \over 1 + ( y^2 / 2) v} \Big\} S \Big( {y^2 \over 2} v \Big) dv 
$$
for $k\in \{1/2, 3/2, 5/2\}$. 
We here use the following asymptotic evaluation of the integral $I(y,k)$:
\begin{equation}\label{I}
\lim_{|y|\to\infty}I(y,k)/S(y^2/2) = \Gamma(b+k),
\end{equation}
for which the proof is given later. 
Using this result, (\ref{post-MSE}) can be approximated as 
\begin{align}
{1 \over m(y)} {\partial ^2 m(y) \over {\partial y}^2} / \frac{2}{y^2} 
&=
\frac{ 2 \Gamma(b+5/2) \{ 1 + o(1) \} - \Gamma(b+3/2) \{ 1 + o(1) \} }{ \Gamma(b+1/2) \{ 1 + o(1) \} } \non \\
&\sim 
(1+b)(1+2b) \non
\end{align}
as $|y| \to \infty $, which is the desired result. 

Finally, we give the proof of (\ref{I}).
Let $M=y^2/2$ and define $h_M(v,k)$ by 
$$
h_M (v,k) = \Big( {M \over 1 + M v} \Big) ^k v^{- b - 1} \exp \Big( - {M \over 1 + M v} \Big) {S(M v) \over S(M)}. 
$$
Then it holds that $I(y,k)/S(y^2/2)=\int_0^{\infty}h_M(v, k)dv$.
Note that 
\begin{align}
{S(M v) \over S(M)} = v^{a + b} \Big( {1 + M \over 1 + M v} \Big) ^{a + b} \Big\{ {1 + \log (1 + M) \over 1 + \log (1 + Mv)} \Big\} ^{1 + \ga } \text{.} \non 
\end{align}
Then, for any $M \ge 1$ and $v \ge 1$, we have 
\begin{align}
h_M (v, k) &\le v^{- k - b - 1} v^{a + b} \Big( {1 + M \over 1 + M v} \Big) ^{a + b} \Big\{ {1 + \log (1 + M) \over 1 + \log (1 + Mv)} \Big\} ^{1 + \ga } \le 2^{a + b} v^{- k - b - 1} \text{.} \non 
\end{align}
Next, for any $M \ge 1$ and $v \le 1$, 
\begin{align}
{1 + \log (1 + M) \over 1 + \log (1 + Mv)} &= \exp \Big\{ \int_{v}^{1} {M \over 1 + M t} {1 \over 1 + \log (1 + M t)} dt \Big\} \non \\
&\le \exp \Big( \int_{v}^{1} {1 \over 1 / M + t} dt \Big) = {1 / M + 1 \over 1 / M + v} \le {2 \over 1 / M + v}. \non 
\end{align}
Then, it follows that, for $M \ge 1$ and $v \le 1$,  
\begin{align}
h_M (v, k) &= {e^{- 1 / (1 / M + v)} \over (1 / M + v)^k} v^{a - 1} \Big( {1 / M + 1 \over 1 / M + v} \Big) ^{a + b} \Big\{ {1 + \log (1 + M) \over 1 + \log (1 + Mv)} \Big\} ^{1 + \ga } \non \\
&\le {e^{- 1 / (1 / M + v)} \over (1 / M + v)^{k + a + b}} v^{a - 1} 2^{a + b} {2^{1 + \ga } \over (1 / M + v)^{1 + \ga }} \non \\
&\le \Big( \sup_{x \in (0, \infty )} {e^{- 1 / x} \over x^{k + a + b + 1 + \ga }} \Big) 2^{a + b + 1 + \ga } v^{a - 1} < \infty \non 
\end{align}
noting that the function $e^{- 1 / x} / x^{k + a + b + 1 + \ga }$ is bounded in $(0, \infty )$. 
Therefore, from the dominated convergence theorem, we have 
$$
\lim_{M\to\infty}\int_0^{\infty}h_M(v,k)
=
\int_0^{\infty}v^{-k-b-1}\exp(-1/v)dv
=
\Gamma(b+k), 
$$
which proves (\ref{I}).

\section{Details on sampling from $\gamma$ given in Algorithm \ref{algo:gamma}}
\label{app:gamma}


We describe and justify the algorithm of the independent Metropolis-Hastings method for sampling hyperparameter $\ga$ by evaluating the upper and lower bounds of the intractable normalizing constant. 
The normalizing constant of the LAS prior, which is dependent on $\gamma$, is given by the integral 
\begin{equation}
C( \ga ) = C \Big( a = {1 \over n}, b = 0, \gamma \Big) = \int _0^1 \frac{\kappa ^{-1} (1-\kappa )^{1 / n - 1} }{ (1-\log \kappa )^{1+\gamma} } d\kappa = \int_{0}^{\infty } g(x; \ga ) dx \text{,} \label{eq:nc_g} 
\end{equation}
where $g(x; \ga ) = (1 - e^{- x} )^{1 / n - 1} (1 + x)^{- 1 - \ga }$ for $x > 0$; the last integral is obtained by the change of variables $\ka = e^{-x}$. 
This is bounded above and below by, with any $K>0$ and $N = K^3$, 
\begin{equation*}
\begin{split}
U( \gamma ,K ) &= \frac{ ( 1-e^{-1/K} )^{1 / n} }{(1 / n) e^{-1/K}} + \frac{(1-e^{-K})^{1 / n - 1}}{\gamma (1+K)^{\gamma}} \\
& \ \ \ \ \ 
+ \sum _{j = 1}^N  \frac{K^2-1}{KN} g \Big( \frac{1+(j - 1)(K^2-1)/N}{K} ; \ga \Big) \\
L( \gamma ,K ) &= \frac{ ( 1-e^{-1/K} )^{1 / n} }{(1 / n) (1+1/K)^{\gamma} } + \frac{1}{\gamma (1+K)^{\gamma}} + \sum _{j = 1}^N \frac{K^2-1}{KN} g \Big( \frac{1 + j (K^2-1)/N}{K} ; \ga \Big) 
\end{split}
\end{equation*}
i.e., $L( \gamma , K) \le C( \gamma ) \le U( \gamma , K)$ for any $( \gamma ,K)$. 
In addition, these bounds can be as tight as desired if one increases $K$; we prove $L( \gamma , K) \to C( \gamma)$ and $U( \gamma , K) \to C( \gamma )$ as $K\to \infty$ in Lemma \ref{lem:Riemann} at the end of this section. 
These bounds are utilized in implementing the independent Metropolis-Hasting algorithm, where the acceptance probability is dependent on the intractable normalizing constant and cannot be directly computed, but their upper and lower bounds are available with arbitrary accuracy. 

The prior for $\ga$ is the gamma distribution, $\gamma \sim Ga(a_0^{\gamma },b_0^{\gamma })$. Each likelihood $\pi ( u_i ; \gamma )$ can be approximated by $\gamma ^a \exp \{ -\gamma \log (1+\log (1+u_i) ) \}$ with $a>0$ so that the gamma prior becomes conjugate. The proposal distribution is the posterior distribution with the approximate likelihoods, or $\gamma \sim Ga(\gamma |a_1^{\gamma},b_1^{\gamma})$, where
\begin{equation*}
a_1^{\gamma} = a_0^{\gamma} + n a, \ \ \ \mathrm{and} \ \ \ b_1^{\gamma} = b_0^{\gamma} + \sum _{i=1}^n \log (1+\log (1+u_i) ),
\end{equation*}
and we set $a = 1 / n$. 
Denote the current state by $\gamma$, and the candidate drawn from the proposal by $\gamma '$. The acceptance probability is 
\begin{equation*}
A(\gamma \to \gamma ') = \Big\{ \frac{C( \gamma )}{C( \gamma ')} \Big\} ^n \left( \frac{\gamma}{\gamma '} \right) ^{n (1 / n)} \text{,} 
\end{equation*}
which is not evaluated directly due to the intractable constant $C( \gamma ) / C( {\ga }' )$. 
We bound these constants to obtain the upper and lower bounds of the acceptance probability. The bounds of the acceptance probability are defined by 
\begin{equation*}
\underline{w} ( \ga , {\ga }' , K) = \Big\{ \frac{L( \gamma , K)}{U( \gamma ' , K)} \Big\} ^n \frac{\gamma}{\gamma '} \qquad \mathrm{and} \qquad \overline{w} ( \ga , {\ga }' , K) = \Big\{ \frac{U( \gamma , K)}{L( \gamma ' , K)} \Big\} ^n \frac{\gamma}{\gamma '} \text{,} 
\end{equation*} 
and satisfy 
\begin{align}
\underline{w} ( \ga , {\ga }' , K) \le A(\gamma \to \gamma ') \le \overline{w} ( \ga , {\ga }' , K). \non 
\end{align}
The definitions of $\underline{w} ( \ga , {\ga }' , K)$ and $\overline{w} ( \ga , {\ga }' , K)$ above are used in the sampling algorithm given in Algorithm \ref{algo:gamma}. 

The bounds, $U(\ga , K)$ and $L(\ga , K)$, are obtained by a straightforward application of the Riemann approximation, and their properties are verified by the following lemma. 
\begin{lem}
\label{lem:Riemann}
Let $g( \cdot )$, $\underline{g} _K ( \cdot )$, and $\overline{g} _K ( \cdot )$, $K = 1, 2, \dotsc $, be integrable functions defined on $(0, \infty )$ satisfying $0 \le \underline{g} _K (x) \le g(x) \le \overline{g} _K (x) < \infty $ for all $x \in (0, \infty )$. 
Assume that $g( \cdot )$ is nonincreasing on $(0, \infty )$. 
Let $0 < x_{0}^{(K)} < \dots < x_{l_K}^{(K)} < \infty $ for $K = 1, 2, \dotsc $ and assume that $\lim_{K \to \infty } x_{0}^{(K)} = 0$ and that $\lim_{K \to \infty } x_{l_K}^{(K)} = \infty $. 
Suppose that $\lim_{K \to \infty } \int_{0}^{x_{0}^{(K)}} \overline{g} _K (x) dx = \lim_{K \to \infty } \int_{x_{l_K}^{(K)}}^{\infty } \overline{g} _K (x) dx = 0$ and that $\lim_{K \to \infty } \sum_{j = 1}^{l_K} ( x_{j}^{(K)} - x_{j - 1}^{(K)} ) \{ g( x_{j - 1}^{(K)} ) - g( x_{j}^{(K)} ) \} = 0$. 
Then 
\begin{align}
0 &\le \int_{0}^{x_{0}^{(K)}} \underline{g} _K (x) dx + \sum_{j = 1}^{l_K} ( x_{j}^{(K)} - x_{j - 1}^{(K)} ) g( x_{j}^{(K)} ) + \int_{x_{l_K}^{(K)}}^{\infty } \underline{g} _K (x) dx \non \\
&\le \int_{0}^{\infty } g(x) dx \non \\
&\le \int_{0}^{x_{0}^{(K)}} \overline{g} _K (x) dx + \sum_{j = 1}^{l_K} ( x_{j}^{(K)} - x_{j - 1}^{(K)} ) g( x_{j - 1}^{(K)} ) + \int_{x_{l_K}^{(K)}}^{\infty } \overline{g} _K (x) dx < \infty \label{eq:Riemann_1} 
\end{align}
for all $K = 1, 2, \dotsc $ and 
\begin{align}
\int_{0}^{\infty } g(x) dx &= \lim_{K \to \infty } \Big\{ \int_{0}^{x_{0}^{(K)}} \underline{g} _K (x) dx + \sum_{j = 1}^{l_K} ( x_{j}^{(K)} - x_{j - 1}^{(K)} ) g( x_{j}^{(K)} ) + \int_{x_{l_K}^{(K)}}^{\infty } \underline{g} _K (x) dx \Big\} \non \\
&= \lim_{K \to \infty } \Big\{ \int_{0}^{x_{0}^{(K)}} \overline{g} _K (x) dx + \sum_{j = 1}^{l_K} ( x_{j}^{(K)} - x_{j - 1}^{(K)} ) g( x_{j - 1}^{(K)} ) + \int_{x_{l_K}^{(K)}}^{\infty } \overline{g} _K (x) dx \Big\} \text{.} \label{eq:Riemann_2} 
\end{align}
\end{lem}

\begin{proof}
The inequalities in (\ref{eq:Riemann_1}) are trivial. 
We obtain (\ref{eq:Riemann_2}) since 
\begin{align}
0 &\le \int_{0}^{x_{0}^{(K)}} \overline{g} _K (x) dx + \sum_{j = 1}^{l_K} ( x_{j}^{(K)} - x_{j - 1}^{(K)} ) g( x_{j - 1}^{(K)} ) \non \\
& \hspace{1cm}  + \int_{x_{l_K}^{(K)}}^{\infty } \overline{g} _K (x) dx - \sum_{j = 1}^{l_K} ( x_{j}^{(K)} - x_{j - 1}^{(K)} ) g( x_{j}^{(K)} ) \non \\
&= \sum_{j = 1}^{l_K} ( x_{j}^{(K)} - x_{j - 1}^{(K)} ) \{ g( x_{j - 1}^{(K)} ) - g( x_{j}^{(K)} ) \} + \int_{0}^{x_{0}^{(K)}} \overline{g} _K (x) dx + \int_{x_{l_K}^{(K)}}^{\infty } \overline{g} _K (x) dx \to 0 \non 
\end{align}
as $K \to \infty $ by assumption. 
\end{proof}

In our problem, where function $g$ is given in (\ref{eq:nc_g}), the condition in the lemma, $\lim_{K \to \infty } \sum_{j = 1}^{l_K} ( x_{j}^{(K)} - x_{j - 1}^{(K)} ) \{ g( x_{j - 1}^{(K)} ) - g( x_{j}^{(K)} ) \} = 0$, is satisfied. To see this, observe that  $\lim_{K \to \infty } \{ \max_{1 \le j \le l_K} ( x_{j}^{(K)} - x_{j - 1}^{(K)} ) \} g( x_{0}^{(K)} ) = 0$, and that the grid becomes sufficiently fine when $K \to \infty $.

\section{Properties of iterated logarithmic functions}
\label{app:IL-prop} 
We here give some properties of iterated logarithmic functions related to Figure \ref{fig:log} in the following lemma. 

\begin{lem} \label{lem:fn}
	For $x>1$,
	\begin{enumerate}
		\item $f_L(x) > 1$.
		
		\item $f_L(x)$ is increasing in $x$.
		
		\item $f_{L + 1} (x) < f_{L}(x)$ (decreasing in $L$ at each point $x$).
		
		\item $\displaystyle \lim _{L\to \infty} f_L (x) = 1$. 
	\end{enumerate}
\end{lem}

\begin{proof}
The first and second properties follow immediately from the definition of $f_L$. 
The third property is verified by the inequality $z-1 > \log z$ for $z > 1$. 
To prove the last property, fix $x > 1$ and write $a_L = f_L (x)$. 
Then this sequence is decreasing and bounded below by $1$. 
Therefore, it has a limit $a = \lim _{L\to \infty} a_L$ in $[1, \infty )$. 
Now, by the definition of $f_{L + 1}$, we have $a_{L + 1} = 1 + \log a_{L}$. 
Letting $L \to \infty $, we have $a = 1 + \log a$, which shows that $a = 1$. 
\end{proof}

\section{Proof of Theorem \ref{thm:properties_of_IL}}
\label{app:IL-proof}
As in the proof of Theorem 
\ref{thm:MSE1}, we suppress the subscript $i$. 
Here we write $\pi (u)$ and $\pi ( \ka )$ for $\pi (u; \ga , L)$ and $\pi ( \ka ; \ga , L)$ and use $p( \th )$ 
to denote the marginal 
density of $\th $ under this prior. 
Let 
\begin{align}
S(u) 
= \Big( {u \over 1 + u} \Big) ^{a + b} \Big\{ \prod_{k = 1}^{L - 1} {1 \over f_k (1 + u)} \Big\} {1 \over \{ f_L (1 + u) \} ^{1 + \ga }} \non 
\end{align}
and let $C = \int_{0}^{\infty } u^{- b - 1} S(u) du$, so that $\pi (u) = C^{- 1} u^{- b - 1} S(u)$. 
Then $S(u)$ is a slowly varying function from Lemma \ref{lem:SV}. 
Note that the above definition of $S(u)$ is not identical to that in Section 
\ref{app:proof_of_MSE1}. 
\label{app:thm3} 
Also, let 
\begin{align}
\pi _0 ( \ka ) &= \pi _{0} ( \ka ; L) = {\partial \over \partial \ka } \Big[ {1 \over \{ f_L (1 / \kappa ) \} ^{\ga }} \Big] = {\ga \over \ka } \Big\{ \prod _{k = 1}^{L - 1} {1 \over f_k (1 / \ka )} \Big\} {1 \over \{ f_L (1 / \ka ) \} ^{1 + \ga }} \text{.} \non 
\end{align}
Then we have 
\begin{align}
\int _0^1 \pi _0 ( \ka ) d\kappa = 1 \text{.} \non 
\end{align}

\subsection*{Integrability (the first property in Theorem \ref{thm:prop})}

The inequality, $\ka ^b (1 - \ka )^{a - 1} \pi _0 ( \ka ) \le (1 - \ka )^{a - 1} \pi _0 ( \ka )$, shows the integrability of $\pi (\ka)$ in $(0,1/2)$,
\begin{align}
\int_{0}^{1 / 2} (1 - \ka )^{a - 1} \pi _0 ( \ka ) d\ka &\le \Big\{ \sup_{\ka \in (0, 1 / 2)} (1 - \ka )^{a - 1} \Big\} \int_{0}^{1 / 2} \pi _0 ( \ka ) d\ka < \infty \non 
\end{align}
and the integrability in $(1/2,1)$, 
\begin{align}
\int_{1 / 2}^{1} (1 - \ka )^{a - 1} \pi _0 ( \ka ) d\ka &\le \int_{1 / 2}^{1} (1 - \ka )^{a - 1} 2 \ga d\ka < \infty \text{.} \non 
\end{align}

\subsection*{Spike around the origin  (the second property in Theorem \ref{thm:prop})}

By the monotone convergence theorem, 
\begin{align}
C \sqrt{2 \pi } p( \th ) &= \int_{0}^{\infty } {u^{- 1 / 2 + a - 1} \over (1 + u)^{a + b}} \Big\{ \prod_{k = 1}^{L - 1} {1 \over f_k (1 + u)} \Big\} {e^{- \th ^2 / 2 u} \over \{ f_L (1 + u) \} ^{1 + \ga }} du \non \\
&\ge \int_{0}^{1} {u^{- 1 / 2 + a - 1} \over 2^{a + b}} \Big\{ \prod_{k = 1}^{L - 1} {1 \over f_k (2)} \Big\} {e^{- \th ^2 / 2 u} \over \{ f_L (2) \} ^{1 + \ga }} du \to \infty \label{tMSE1p1} 
\end{align}
as $\th \to 0$ for $a \le 1 / 2$. 

\subsection*{Slowly varying density  (the third property in Theorem \ref{thm:prop})}

By the change of variables $u = ( \th ^2 / 2) v$, it follows that 
\begin{align}
p( \th ) &= {C^{- 1} \over \sqrt{2 \pi }} \int_{0}^{\infty } u^{- 1 / 2 - b - 1} e^{- ( \th ^2 / 2) / u} S(u) du \non \\
&= {C^{- 1} \over \sqrt{2 \pi }} \Big( {\th ^2 \over 2} \Big) ^{- 1 / 2 - b} \int_{0}^{\infty } v^{- 1 / 2 - b - 1} e^{- 1 / v} S \Big( {\th ^2 \over 2} v \Big) dv \text{.} \non 
\end{align}
Now for $\th ^2 / 2 \ge 1$, we have 
\begin{align*}
{S(( \th ^2 / 2) v) \over S( \th ^2 / 2)} &= v^{a + b} \Big\{ {1 + \th ^2 / 2 \over 1 + ( \th ^2 / 2) v} \Big\} ^{a + b} \non \\
&\quad \times \Big\{ \prod_{k = 1}^{L - 1} {f_{k} (1 + \th ^2 / 2) \over f_{k} (1 + ( \th ^2 / 2) v)} \Big\} \Big\{ {f_L (1 + \th ^2 / 2) \over f_L (1 + ( \th ^2 / 2) v) } \Big\} ^{1 + \ga } \\
&\le 2^{a + b} \max \{ 1, 1 / v^{L + \ga } \} \non 
\end{align*}
since 
\begin{align}
{1 + \th ^2 / 2 \over 1 + ( \th ^2 / 2) v} &\le {1 \over v} {1 + \th ^2 / 2 \over \th ^2 / 2} \le {2 \over v} \non 
\end{align}
and since for all $k = 1, \dots , L$, 
\begin{align}
{f_{k} (1 + \th ^2 / 2) \over f_{k} (1 + ( \th ^2 / 2) v)} \le 1 \non 
\end{align}
when $v \ge 1$ while 
\begin{align}
{f_{k} (1 + \th ^2 / 2) \over f_{k} (1 + ( \th ^2 / 2) v)} 
&= \exp \Big\{ \int_{t = v}^{t = 1} {\partial \over \partial t} \log f_{k} (1 + ( \th ^2 / 2) t) dt \Big\} \non \\
&= \exp \Big\{ \int_{t = v}^{t = 1} {1 \over f_{k} (1 + ( \th ^2 / 2) t) \dotsm f_{1} (1 + ( \th ^2 / 2) t)} {\th ^2 / 2 \over 1 + ( \th ^2 / 2) t} dt \Big\} \non \\
&\le \exp \Big( \int_{t = v}^{t = 1} {dt \over t} \Big) = 1 / v \non 
\end{align}
when $v < 1$. 
Thus, by the dominated convergence theorem, we obtain 
\begin{align}
{p( \th ) \over S( \th ^2 / 2)} \sim {C^{- 1} \over \sqrt{2 \pi }} \Big( {\th ^2 \over 2} \Big) ^{- 1 / 2 - b} \Ga \Big( {1 \over 2} + b \Big) \non 
\end{align}
as $| \th | \to \infty $, where $S( \th ^2 / 2)$ is a slowly varying function of $| \th |$. This completes the proof of the first statement of Theorem~\ref{thm:properties_of_IL}. 

\subsection*{Convergence to the degenerate distribution (the second property in Theorem \ref{thm:properties_of_IL})}

To prove the second statement of Theorem~\ref{thm:properties_of_IL}, we have that 
\begin{align}
\lim_{L \to \infty } \int_{0}^{\ep } \pi _0 ( \ka ; L) d\ka = \lim_{L \to \infty } \Big[ {1 \over \{ f_L (1 / \ka ) \} ^{\ga }} \Big] _{0}^{\ep } = \lim_{L \to \infty } {1 \over \{ f_L (1 / \ep ) \} ^{\ga }} = 1 \label{plim_newp1} 
\end{align}
for all $0 < \ep < 1$. Hence, 
\begin{align}
0 &= \limsup_{L \to \infty } \int_{\ka }^{1} \pi _0 ( \tilde{\ka } ; L) d\tilde{\ka } \ge \limsup_{L \to \infty } \int_{\ka }^{1} \ka \pi _0 ( \ka ; L) d\tilde{\ka } = \limsup_{L \to \infty } \{ \ka (1 - \ka ) \pi _0 ( \ka ; L) \} \ge 0 \text{,} \non 
\end{align}
or, equivalently, $\lim_{L \to \infty } \pi _0 ( \ka ; L) = 0$, for all $\ka \in (0, 1)$. 
Now, set $b = 0$. 
Then, for any $0 < \ep < 1$, the probability of $( \ep , 1)$ under the prior $\pi ( \ka ; L) = \pi ( \ka ; \ga , L) \propto (1 - \ka )^{a - 1} \pi _0 ( \ka ; L)$ is 
\begin{align}
\int_{\ep }^{1} \pi ( \ka ; L) d\ka &= \frac{\displaystyle \int_{\ep }^{1} (1 - \ka )^{a - 1} \pi _0 ( \ka ; L) d\ka  }{\displaystyle \int_{0}^{1} (1 - \ka )^{a - 1} \pi _0 ( \ka ; L) d\ka  } \le \frac{\displaystyle \int_{\ep }^{1} (1 - \ka )^{a - 1} \pi _0 ( \ka ; L) d\ka }{\displaystyle \Big\{ \inf_{\ka \in (0, 1 / 2)} (1 - \ka )^{a - 1} \Big\} \int_{0}^{1 / 2} \pi _0 ( \ka ; L) d\ka  } \text{.} \non 
\end{align}
The numerator on the right side 
converges to zero as $L \to \infty $ by the dominated convergence theorem since $\pi _0 ( \ka ; L) \le \ga / \ep$ for all $\ep < \ka < 1$ and $\pi _0 ( \ka ; L) \to 0$. 
Also, it follows from (\ref{plim_newp1}) that $\int_{0}^{1 / 2} \pi _0 ( \ka ; L) d\ka \to 1$ as $L \to \infty $. 
Thus, $\lim_{L \to \infty } \int_{0}^{\ep } \pi ( \ka ; L) d\ka = \lim_{L \to \infty } \{ 1 - \int_{\ep }^{1} \pi ( \ka ; L) d\ka \} = 1$, which is the desired result.

\section{Proof of Theorem \ref{thm:MSE2}}
\label{app:thm4} 

As in the 
previous section, we suppress the subscript $i$. 
Let $\pi (u)$, 
$S(u)$, and $C$ be as in Section \ref{app:thm3}. 

%
The formal proof of Theorem \ref{thm:MSE2} is completely analogous to that of Theorem \ref{thm:MSE1} since 
\begin{align}
&- {3 a / 2 \over y^2 / 2} + \sum_{k - 1}^{L - 1} {1 \over f_k (1 + y^2 / 2) \dotsm f_1 (1 + y^2 / 2)} + {1 + \ga \over f_L (1 + y^2 / 2) \dotsm f_1 (1 + y^2 / 2)} = o(1) \non 
\end{align}
as $|y| \to \infty $. 
In the following, we informally derive the expression (\ref{eq:MSE-IL}) using integration by parts. 

First, $\pi (u)$ is approximated by $\pit (u) = \pi (u) e^{- \ep / u}$ for some small $\ep > 0$ in the sense that 
\begin{align}
{\rm{MSE}}_{\pi } ( \th | y) - 1 \approx {\rm{MSE}}_{\pit } ( \th | y) - 1 \text{.} \label{pMSE-ILp00} 
\end{align}
Denote by $\tilde{m} (y)$ the marginal density of $y$ under the prior ${\tilde{C}}^{- 1} \pit (u)$, where $\tilde{C} = \int_{0}^{\infty } \pit (u) du$. 
Then, as in the proof of Theorem \ref{thm:MSE1}, we have that 
$$
{\rm{MSE}}_{\pit } ( \th | y) = 1 + {{\tilde{m}}'' (y) \over \tilde{m} (y)} 
$$
and that 
$$
\tilde{m} (y) = {{\tilde{C}}^{- 1} \over \sqrt{2 \pi }} \int_{0}^{\infty } {1 \over \sqrt{1 + u}} \exp \Big( - {y^2 / 2 \over 1 + u} \Big) \pit (u) du \text{.} 
$$
Now, by integration by parts, 
\begin{align}
{\tilde{m}}'' (y) &= {{\tilde{C}}^{- 1} \over \sqrt{2 \pi }} \int_{0}^{\infty } \Big\{ {y^2 \over (1 + u)^2} - {1 \over 1 + u} \Big\} {1 \over \sqrt{1 + u}} \exp \Big( - {y^2 / 2 \over 1 + u} \Big) \pit (u) du \non \\
&= {{\tilde{C}}^{- 1} \over \sqrt{2 \pi }} \Big\{ \Big[ {2 \over \sqrt{1 + u}} \exp \Big( - {y^2 / 2 \over 1 + u} \Big) \pit (u) \Big] _{0}^{\infty } \non \\
&\quad - \int_{0}^{\infty } {2 \over \sqrt{1 + u}} \exp \Big( - {y^2 / 2 \over 1 + u} \Big) {\pit }' (u) du \Big\} \non \\
&= - 2 {{\tilde{C}}^{- 1} \over \sqrt{2 \pi }} \int_{0}^{\infty } {1 \over \sqrt{1 + u}} \exp \Big( - {y^2 / 2 \over 1 + u} \Big) {\pit }' (u) du \text{.} \non 
\end{align}
Therefore, by the change of variables $u = ( y^2 / 2) v$, we obtain 
\begin{align}
{\rm{MSE}}_{\pit } ( \th | y) 
&= 1 - 2 \frac{ \displaystyle \int_{0}^{\infty } {1 \over \sqrt{1 + u}} \exp \Big( - {y^2 / 2 \over 1 + u} \Big) {\pit }' (u) du }{ \displaystyle \int_{0}^{\infty } {1 \over \sqrt{1 + u}} \exp \Big( - {y^2 / 2 \over 1 + u} \Big) \pit (u) du } \non \\
&= 1 - 2 \frac{ \displaystyle \int_{0}^{\infty } \sqrt{{y^2 / 2 \over 1 + ( y^2 / 2) v}} \exp \Big\{ - {y^2 / 2 \over 1 + ( y^2 / 2) v} \Big\} {\pit }' \Big( {y^2 \over 2} v \Big) dv }{ \displaystyle \int_{0}^{\infty } \sqrt{{y^2 / 2 \over 1 + ( y^2 / 2) v}} \exp \Big\{ - {y^2 / 2 \over 1 + ( y^2 / 2) v} \Big\} \pit \Big( {y^2 \over 2} v \Big) dv } \text{.} \label{pMSE-ILp1} 
\end{align}
Now, if $|y|$ is sufficiently large, it follows that 
\begin{align}
&\sqrt{{y^2 / 2 \over 1 + ( y^2 / 2) v}} \exp \Big\{ - {y^2 / 2 \over 1 + ( y^2 / 2) v} \Big\} \approx {1 \over \sqrt{v}} \exp \Big( - {1 \over v} \Big) \label{pMSE-ILp2} 
\end{align}
and that 
\begin{align}
&\pit \Big( {y^2 \over 2} v \Big) \approx \pi \Big( {y^2 \over 2} v \Big) \approx C^{- 1} \Big( {y^2 \over 2} \Big) ^{- 1} v^{- 1} S \Big( {y^2 \over 2} \Big) \label{pMSE-ILp3} 
\end{align}
since $S(u)$ is a slowly varying function. 
Furthermore, since 
\begin{align}
{{\pit }' (u) \over \pit (u)} - {\ep \over u^2} = {{\pi }' (u) \over \pi (u)} &= - {1 \over u} + {a \over u (1 + u)} - \sum_{k = 1}^{L - 1} {1 \over f_k (1 + u)} \dotsm {1 \over f_1 (1 + u)} {1 \over 1 + u} \non \\
&\quad - (1 + \ga ) {1 \over f_L (1 + u)} \dotsm {1 \over f_1 (1 + u)} {1 \over 1 + u} \non 
\end{align}
and since $f_k (1 + u)$, $k = 1, \dots , L$, are slowly varying functions, it also follows that 
\begin{align}
{{\pit }' (( y^2 / 2) v) \over \pit (( y^2 / 2) v)} &\approx - {1 \over y^2 / 2} {1 \over v} + {a\over ( y^2 / 2)^2} {1 \over v^2} - \sum_{k = 1}^{L - 1} {1 \over f_k (1 + y^2 / 2)} \dotsm {1 \over f_1 (1 + y^2 / 2)} {1 \over y^2 / 2} {1 \over v} \non \\
&\quad - (1 + \ga ) {1 \over f_L (1 + y^2 / 2)} \dotsm {1 \over f_1 (1 + y^2 / 2)} {1 \over y^2 / 2} {1 \over v} \label{pMSE-ILp4} 
\end{align}
for sufficiently large $|y|$. 
Substituting (\ref{pMSE-ILp2}), (\ref{pMSE-ILp3}), and (\ref{pMSE-ILp4}) into (\ref{pMSE-ILp1}) yields 
\begin{align}
{\rm{MSE}}_{\pit } ( \th | y) 
&\approx 1 - 2 \frac{ \int_{0}^{\infty } v^{- 1 / 2} e^{- 1 / v} v^{- 1} \{ {\pit }' (( y^2 / 2) v) / \pit (( y^2 / 2) v) \} dv }{ \int_{0}^{\infty } v^{- 1 / 2} e^{- 1 / v} v^{- 1} dv } \non \\
&\approx 1 + 2 \Big\{ {1 \over y^2 / 2} {1 \over 2} - {a \over ( y^2 / 2)^2} {1 \over 2} {3 \over 2} \non \\
&\quad + \sum_{k = 1}^{L - 1} {1 \over f_k (1 + y^2 / 2)} \dotsm {1 \over f_1 (1 + y^2 / 2)} {1 \over y^2 / 2} {1 \over 2} \non \\
&\quad + (1 + \ga ) {1 \over f_L (1 + y^2 / 2)} \dotsm {1 \over f_1 (1 + y^2 / 2)} {1 \over y^2 / 2} {1 \over 2} \Big\} \label{pMSE-ILp5} 
\end{align}
since $\Ga (1 / 2 + k) = \int_{0}^{\infty } v^{- 1 / 2 - k - 1} e^{- 1 / v} dv$ for $k = 0, 1, 2$. 
Finally, combining (\ref{pMSE-ILp00}) and (\ref{pMSE-ILp5}) yields (\ref{eq:MSE-IL}).

\section{Derivation of the augmentation (\ref{eq:augmentation_IL}) and Gibbs sampling given in Algorithm \ref{algo:gibbs-ILAS}}
\label{app:gibbs-ILAS} 

Let $u > 0$, $c_0 > 0$, and $c_1 , \dots , c_L \ge 0$. 
We first note that, for any positive $(t_0,t_1,\ldots , t_L)$, 
\begin{align}
&{1 \over (1 + u)^{c_L}} \prod_{k = 1}^{L} {1 \over \{ f_k (1 + u) \} ^{c_{L - k}}} \non \\
& \ \ \ = {{\rm{Ga}} ( t_0 | c_0 , 1) \over {\rm{Ga}} ( t_0 | c_0 , f_L (1 + u))} {{\rm{Ga}} ( t_1 | t_0 + c_1 , 1) \over {\rm{Ga}} ( t_1 | t_0 + c_1 , f_{L - 1} (1 + u))} \times \dotsm \non \\
& \ \ \ \  \times {{\rm{Ga}} ( t_{L - 1} | t_{L - 2} + c_{L - 1} , 1) \over {\rm{Ga}} ( t_{L - 1} | t_{L - 2} + c_{L - 1} , f_1 (1 + u))} {{\rm{Ga}} ( t_L | t_{L - 1} + c_L , 1) \over {\rm{Ga}} ( t_L | t_{L - 1} + c_L , 1 + u)} e^{- t_L u} \text{.} \label{eq:GS_1} 
\end{align}
The denominator in the right is in fact the probability density of $t_{0:L} \in (0,\infty )^{L+1}$; 
\begin{align}
\int_{(0, \infty )^{L + 1}} {\rm{Ga}} ( t_0 | c_0 , f_L (1 + u)) 
\dotsm {\rm{Ga}} ( t_L | t_{L - 1} + c_L , 1 + u) d{t_{0:L}} = 1 \text{.} \non 
\end{align}
From (\ref{eq:GS_1}), we obtain 
\begin{align}
&{1 \over (1 + u)^{c_L}} \prod_{k = 1}^{L} {1 \over \{ f_k (1 + u) \} ^{c_{L - k}}} \non \\
&= \int_{(0, \infty )^{L + 1}} \Big\{ {\rm{Ga}} ( t_0 | c_0 , 1) {\rm{Ga}} ( t_1 | t_0 + c_1 , 1) \times \dotsm \non \\
&\quad \times {\rm{Ga}} ( t_{L - 1} | t_{L - 2} + c_{L - 1} , 1) {\rm{Ga}} ( t_L | t_{L - 1} + c_L , 1) e^{- t_L u} \Big\} d{t_{0:L}} \label{eq:integration_GS_gamma} 
\end{align}
Expression (\ref{eq:integration_GS_gamma}) gives the augmentation (\ref{eq:augmentation_IL}) by setting 
$c_0 = 1 + \ga $, $c_1=\cdots = c_{L-1} =1$, and $c_L=a+b$. 
Moreover, the full conditional distributions of the latent variables in Algorithm \ref{algo:gibbs-ILAS} are obtained from (\ref{eq:augmentation_IL}) and (\ref{eq:GS_1}). 
\section{Properties of doubly log-adjusted shrinkage priors in Section \ref{sec:discussion}}
\label{app:double-IL}

In this section, we prove the results stated in the discussion of Section \ref{sec:discussion}. 
For $\al , \be > 0$, consider the density of the doubly log-adjusted prior given by 
\begin{align}
\pi ( \ka ; \al , \be , L) \propto \pi _0 ( \ka ; \al , L) \pi _0 (1 - \ka ; \be , L) \text{,} \quad \ka \in (0, 1) \text{,} \non 
\end{align}
where 
\begin{align}
\pi _0 ( \ka ; \ga , L) &= {\partial \over \partial \ka } \Big[ {1 \over \{ f_L (1 / \kappa ) \} ^{\ga }} \Big] = {\ga \over \ka } \Big\{ \prod _{k = 1}^{L - 1} {1 \over f_k (1 / \ka )} \Big\} {1 \over \{ f_L (1 / \ka ) \} ^{1 + \ga }} \non 
\end{align}
as defined in Appendix~\ref{app:thm3} (or, this is the ILAS prior with $a = 1$ and $b = 0$ in (\ref{ILAS})), and let $F( \ka ; \al , \be , L)$ denote the corresponding distribution function. 
For $0 < \ep < 1$, let 
\begin{align}
R( \ep ; \al , \be , L) = {F( \ep ; \al , \be , L) \over 1 - F(1 - \ep ; \al , \be , L)} \non 
\end{align}
be the ratio of the prior probability of $\ka \in (0, \ep )$ to that of $\ka \in (1 - \ep , 1)$. 
Proposition \ref{prp:double-IL} summarizes the fundamental properties of $\pi ( \ka ; \al , \be , L)$. 

\begin{prp}
\label{prp:double-IL}
The prior $\pi ( \ka ; \al , \be , L)$ satisfies the following properties. 
\begin{enumerate}
\item
$\int _{0}^{1} \pi ( \ka ; \al , \be , L) d\ka < \infty $. 
\item
\begin{enumerate}

\item
If $0 < \ep \le 1 / 2$, then $R( \ep ; \al , \be , L)$ is increasing in $\be $ and decreasing in $\al $. 
\item
$\lim_{\be \to 0} R( \ep ; \al , \be , L) = 0$ and $\lim_{\al \to 0} R( \ep ; \al , \be , L) = \infty $. 
\item 
$R( \ep ; \al , \be , L) \gtreqless 1$ if and only if $\al \lesseqgtr \be $. 
\end{enumerate}
\item
For two arbitrary bounded sequences of positive real numbers, $\al _L$ and $\be _L$, $L = 1, 2, \dotsc $, we have $\lim_{L \to \infty } \{ F(1 - \ep ; \al _L , \be _L , L) - F( \ep ; \al _L , \be _L , L) \} = 0$ if $0 < \ep < 1 / 2$. 
\item
The prior density of $u = (1 - \ka ) / \ka $ can be expressed as 
\begin{align}
\pi (u; \al , \be , L) &\propto 
{1 \over u} \int_{(0, \infty )^{L + 1}} \Big\{ {\rm{Ga}} ( r_0 | 1 + \al , 1) {\rm{Ga}} ( r_1 | r_0 + 1, 1) \times \dotsm \non \\
&\quad \times {\rm{Ga}} ( r_{L - 1} | r_{L - 2} + 1, 1) {\rm{Ga}} ( r_L | r_{L - 1} , 1) e^{- r_L u} \Big\} d{r_{0:L}} \non \\
&\quad \times \int_{(0, \infty )^{L + 1}} \Big\{ {\rm{Ga}} ( s_0 | 1 + \be , 1) {\rm{Ga}} ( s_1 | s_0 + 1, 1) \times \dotsm \non \\
&\quad \times {\rm{Ga}} ( s_{L - 1} | s_{L - 2} + 1, 1) {\rm{Ga}} ( s_L | s_{L - 1} , 1) e^{- s_L / u} \Big\} d{s_{0:L}} \text{.} \non 
\end{align}

\end{enumerate}
\end{prp}

\begin{proof}
Let $\pit _0 ( \ka ; \ga , L) = \pi _0 ( \ka ; \ga , L) / \ga $. 
Then property 1 follows immediately by 
\begin{align}
&\int_{0}^{1} \pit _0 ( \ka ; \al , L) \pit _0 (1 - \ka ; \be , L) d\ka \non \\
&\le \int_{0}^{1} \pit _0 ( \ka ; \min \{ \al , \be \} , L) \pit _0 (1 - \ka ; \min \{ \al , \be \} , L) d\ka \non \\
&= 2 \int_{0}^{1 / 2} \pit _0 ( \ka ; \min \{ \al , \be \} , L) \pit _0 (1 - \ka ; \min \{ \al , \be \} , L) d\ka \non \\
&\le 4 \int_{0}^{1 / 2} \pit _0 ( \ka ; \min \{ \al , \be \} , L) d\ka < \infty \text{.} \non 
\end{align}

For property 2, we start from the proof of 2-(a) and 2-(b) with the focus on $\be$. Note that 
\begin{align}
R( \ep ; \al , \be , L) &= \frac{ \int_{0}^{\ep } \pit _0 ( \ka ; \al , L) \pit _0 (1 - \ka ; \be , L) d\ka }{ \int_{1 - \ep }^{1} \pit _0 ( \ka ; \al , L) \pit _0 (1 - \ka ; \be , L) d\ka } = \frac{ \int_{0}^{\ep } \pit _0 ( \ka ; \al , L) \pit _0 (1 - \ka ; \be , L) d\ka }{ \int_{0}^{\ep } \pit _0 ( \ka ; \be , L) \pit _0 (1 - \ka ; \al , L) d\ka } \text{.} \non 
\end{align}
Then we have 
\begin{align}
&\Big\{ \int_{0}^{\ep } \pit _0 ( \ka ; \be , L) \pit _0 (1 - \ka ; \al , L) d\ka \Big\} ^2 {\partial \over \partial \be } R( \ep ; \al , \be , L) \non \\
&= \int_{0}^{\ep } \pit _0 ( \ka ; \al , L) \pit _0 (1 - \ka ; \be , L) \{ - \log f_L (1 / (1 - \ka )) \} d\ka \int_{0}^{\ep } \pit _0 ( \ka ; \be , L) \pit _0 (1 - \ka ; \al , L) d\ka \non \\
&\quad - \int_{0}^{\ep } \pit _0 ( \ka ; \al , L) \pit _0 (1 - \ka ; \be , L) d\ka \int_{0}^{\ep } \pit _0 ( \ka ; \be , L) \pit _0 (1 - \ka ; \al , L) \{ - \log f_L (1 / \ka ) \} d\ka \non \\
&> [ \{ - \log f_L (1 / (1 - \ep )) \} - \{ - \log f_L (1 / \ep ) \} ] \non \\
&\quad \times \int_{0}^{\ep } \pit _0 ( \ka ; \al , L) \pit _0 (1 - \ka ; \be , L) d\ka \int_{0}^{\ep } \pit _0 ( \ka ; \be , L) \pit _0 (1 - \ka ; \al , L) d\ka \text{.} \non 
\end{align}
The right-hand side is nonnegative for $0 < \ep \le 1 / 2$. 
Thus, $R( \ep ; \al , \be , L)$ is increasing in $\be $ for $0 < \ep \le 1 / 2$. 
In addition, 
\begin{align}
0 &\le R( \ep ; \al , \be , L) = \frac{ \int_{0}^{\ep } \pi _0 ( \ka ; \al , L) \pi _0 (1 - \ka ; \be , L) d\ka }{ \int_{0}^{\ep } \pi _0 ( \ka ; \be , L) \pi _0 (1 - \ka ; \al , L) d\ka } \non \\
&\le {\be / \al \over 1 - \ep } \Big\{ \prod_{k = 1}^{L - 1} f_k (1 / (1 - \ep )) \Big\} \{ f_L (1 / (1 - \ep )) \} ^{1 + \al } \frac{ \int_{0}^{\ep } \pi _0 ( \ka ; \al , L) d\ka }{ \int_{0}^{\ep } \pi _0 ( \ka ; \be , L) d\ka } \non \\
&= \Big\{ \prod_{k = 1}^{L - 1} f_k (1 / (1 - \ep )) \Big\} \{ f_L (1 / (1 - \ep )) \} ^{1 + \al } \frac{ \int_{0}^{\ep } \pi _0 ( \ka ; \al , L) d\ka }{\al (1 - \ep ) } {\be \over \{ f_L (1 / \ep ) \} ^{- \be } } \to 0 \non 
\end{align}
as $\be \to 0$. 

To prove 2-(a) and 2-(b) for $\al$, note that $R( \ep ; \al , \be , L) = 1 / R( \ep ; \be , \al , L)$ for any $\al , \be > 0$ and any $0 < \ep < 1$. From this fact, it is immediate that $R( \ep ; \al , \be , L) \to \infty $ as $\al \to 0$ and that $R( \ep ; \al , \be , L)$ is decreasing in $\al $ for $0 < \ep \le 1 / 2$. 

To prove property 2-(c), we first assume $0< \ep \le 1/2$. Because $R(\ep; \al , \al ; L) = 1$ for any $\ep \in (0,1)$ by definition and it is increasing in the second $\al$, we have $1 = R(\ep; \al , \al ; L) < R(\epsilon ; \al ,\be, L)$ for $\al < \be$. Similarly, we have $R(\epsilon ; \al ,\be, L) < 1$ for $\al > \be$. To extend this result to $1/2 < \epsilon < 1$, we first confirm that $R(1-\ep , \al , \be , L) > 1$ for $\al < \be$, which implies $\int_{0}^{1 - \ep } \pi _0 ( \ka ; \al , \be , L) d\ka > \int_{0}^{1 - \ep } \pi _0 ( \ka ; \be , \al , L) d\ka$. Then, observe that  
\begin{align}
R( \ep ; \al , \be , L) &= \frac{ \int_{0}^{\ep } \pi _0 ( \ka ; \al , \be , L) d\ka }{ \int_{0}^{\ep } \pi _0 ( \ka ; \be, \al , L) d\ka } \non \\
&= \frac{ 1 - \int_{\ep }^{1} \pi _0 ( \ka ; \al , \be , L) d\ka }{ 1 - \int_{\ep }^{1} \pi _0 ( \ka ; \be, \al , L) d\ka } = \frac{ 1 - \int_{0}^{1 - \ep } \pi _0 ( \ka ; \be , \al , L) d\ka }{ 1 - \int_{0}^{1 - \ep } \pi _0 ( \ka ; \al , \be , L) d\ka } > 1 \text{.} \non 
\end{align}
The same argument applies to $\al > \be$. Thus, we conclude for any $0 < \ep < 1$ that $R( \ep ; \al , \be , L) \gtreqless 1$ if and only if $\al \lesseqgtr \be $, which completes the proof of 2-(c). 

For the proof of property 3, let $F_0 ( \ka ; \ga , L)$ denote the distribution function of the prior $\pi _0 ( \ka ; \ga , L)$. Select $M > 0$ large enough so that $0 < \al _L < M$ for all $L \ge 1$. 
Then we have 
\begin{align}
1 &\ge F_0 ( \ep ; \al _L , L) = \int_{0}^{\ep } \pi _0 ( \ka ; \al _L , L) d\ka 
= \{ f_L (1 / \ep ) \} ^{- \al _L } \ge \{ f_L (1 / \ep ) \} ^{- M} \to 1 \label{pdouble-ILp1} 
\end{align}
as $L \to \infty $. 
Next, 
\begin{align}
F(1 - \ep ; \al _L , \be _L , L) - F( \ep ; \al _L , \be _L , L) &= \int_{\ep }^{1 - \ep } \pi ( \ka ; \al _L , \be _L , L) d\ka = {I( \ep ; \al _L , \be _L , L) \over I(0; \al _L , \be _L , L)} \text{,} \non 
\end{align}
where 
\begin{align}
I( \om ; \al , \be , L) &= \int_{\om }^{1 - \om } \pi _0 ( \ka ; \al , L) \pi _0 (1 -  \ka ; \be , L) d\ka \non 
\end{align}
for $\om \in [0, 1]$ and $\al , \be > 0$. 
Now suppose $0 < \ep < 1 / 2$ and let 
$U$ be a uniform random variable on the interval $( \ep , 1 - \ep )$. 
Then it follows from the covariance inequality that 
\begin{align}
{I( \ep ; \al _L , \be _L , L) \over 1 - 2 \ep } &= E[ \pi _0 (U; \al _L , L) \pi _0 (1 - U; \be _L , L) ] \non \\
&= E \Big[ {\al _L \over 1 - U} \Big\{ \prod_{k = 1}^{L - 1} {1 \over f_k (1 / U)} \Big\} {1 \over \{ f_L (1 / U) \} ^{1 + \al _L }} \non \\
&\quad \times {\be _L \over U} \Big\{ \prod_{k = 1}^{L - 1} {1 \over f_k (1 / (1 - U))} \Big\} {1 \over \{ f_L (1 / (1 - U)) \} ^{1 + \be _L }} \Big] \non \\
&\le E \Big[ {\al _L \over 1 - U} \Big\{ \prod_{k = 1}^{L - 1} {1 \over f_k (1 / U)} \Big\} {1 \over \{ f_L (1 / U) \} ^{1 + \al _L }} \Big] \non \\
&\quad \times E \Big[ {\be _L \over U} \Big\{ \prod_{k = 1}^{L - 1} {1 \over f_k (1 / (1 - U))} \Big\} {1 \over \{ f_L (1 / (1 - U)) \} ^{1 + \be _L }} \Big] \non \\
&= \widetilde{I} ( \ep ; \al _L , L) \widetilde{I} ( \ep ; \be _L , L) \text{,} \non 
\end{align}
where 
\begin{align}
\widetilde{I} ( \ep ; \ga , L) &= E \Big[ {\ga \over 1 - U} \Big\{ \prod_{k = 1}^{L - 1} {1 \over f_k (1 / U)} \Big\} {1 \over \{ f_K (1 / U) \} ^{1 + \ga }} \Big] \non \\
&= E \Big[ {\ga \over U} \Big\{ \prod_{k = 1}^{L - 1} {1 \over f_k (1 / (1 - U))} \Big\} {1 \over \{ f_L (1 / (1 - U)) \} ^{1 + \ga }} \Big] \non 
\end{align}
for $\ga > 0$. 
Furthermore, 
\begin{align}
(1 - 2 \ep ) \widetilde{I} ( \ep ; \ga , L) &= \int_{\ep }^{1 - \ep } {\ka \over 1 - \ka } \pi _0 ( \ka ; \ga , L) d\ka \le {1 - \ep \over \ep } \{ F_0 (1 - \ep ; \ga , L) - F_0 ( \ep ; \ga , L) \} \non 
\end{align}
for all $\ga > 0$. 
On the other hand, letting 
\begin{align}
h( \ka ; \ga , L) &= \ga \Big\{ \prod_{k = 1}^{L - 1} {1 \over f_k (1 / \ka )} \Big\} {1 \over \{ f_L (1 / \ka ) \} ^{1 + \ga }} \non 
\end{align}
for $\ka \in (0, 1)$ and $\ga > 0$. Noting that $h$ is increasing in $\ka$, we obtain 
\begin{align}
I(0; \al _L , \be _L , L) &\ge \int_{0}^{\ep } \pi _0 ( \ka ; \al _L , L) \be _L \Big\{ \prod_{k = 1}^{L - 1} {1 \over f_k (1 / (1 - \ep ))} \Big\} {1 \over \{ f_L (1 / (1 - \ep )) \} ^{1 + \be _L }} d\ka \non \\
&= F_0 ( \ep ; \al _L , L) h(1 - \ep ; \be _L , L) = {F_0 ( \ep ; \al _L , L) \over 1 - 2 \ep } \int_{\ep }^{1 - \ep } h(1 - \ep ; \be _L , L) d\ka \non \\
&\ge {F_0 ( \ep ; \al _L , L) \over 1 - 2 \ep } \int_{\ep }^{1 - \ep } {\ep \over \ka } h( \ka ; \be _L , L) d\ka = {F_0 ( \ep ; \al _L , L) \over 1 - 2 \ep } \int_{\ep }^{1 - \ep } \ep \pi _0 ( \ka ; \be _L , L) d\ka \non \\
&= {\ep \over 1 - 2 \ep } F_0 ( \ep ; \al _L , L) \{ F_0 (1 - \ep ; \be _L , L) - F_0 ( \ep ; \be _L , L) \} \text{.} \non 
\end{align}
Thus, we conclude by (\ref{pdouble-ILp1}) that 
\begin{align}
&F(1 - \ep ; \al _L , \be _L , L) - F( \ep ; \al _L , \be _L , L) \non \\
&\le {1 \over 1 - 2 \ep } \Big( {1 - \ep \over \ep } \Big) ^2 {1 - 2 \ep \over \ep } \frac{ \{ F_0 (1 - \ep ; \al _L , L) - F_0 ( \ep ; \al _L , L) \} \{ F_0 (1 - \ep ; \be _L , L) - F_0 ( \ep ; \be _L , L) \} }{ F_0 ( \ep ; \al _L , L) \{ F_0 (1 - \ep ; \be _L , L) - F_0 ( \ep ; \be _L , L) \} } \non \\
&= {(1 - \ep )^2 \over {\ep }^3} \frac{ F_0 (1 - \ep ; \al _L , L) - F_0 ( \ep ; \al _L , L) }{ F_0 ( \ep ; \al _L , L) } \to 0 \non 
\end{align}
as $L \to \infty $. 

For part 4, note that the unnormalized density of $u = (1 - \ka ) / \ka $ based on $\pi _0 ( \ka ; \al , L) \pi _0 (1 - \ka ; \be , L)$ is 
\begin{align}
{\al \be \over u} \Big\{ \prod_{k = 1}^{L - 1} {1 \over f_k (1 + u)} \Big\} {1 \over \{ f_L (1 + u) \} ^{1 + \al }} \Big\{ \prod_{k = 1}^{L - 1} {1 \over f_k (1 + 1 / u)} \Big\} {1 \over \{ f_L (1 + 1 / u) \} ^{1 + \be }} \text{.} \non 
\end{align}
Then, apply the integral representation in (\ref{eq:integration_GS_gamma}) to the two products of functions of $1+u$ and $1+1/u$ 
with $c_L=0$, $c_{L-1}=\cdots = c_1 = 1$ and $c_0\in \{ 1 + \al , 1 + \be \}$ to obtain the desired result. 

\end{proof}

\end{document}